\documentclass[journal,letterpaper]{IEEEtran}
\usepackage{amsmath}
\usepackage{amsthm}
\usepackage{amsfonts}
\usepackage{amssymb}
\usepackage{makeidx}
\usepackage{cite}
\usepackage{graphicx,bigints}
\usepackage{mathtools}
\usepackage{cases}
\usepackage{color}
\usepackage{hyperref}
\usepackage{bbm}
\usepackage[normalem]{ulem}
\usepackage{braket}
\usepackage{textcomp,gensymb}

\usepackage{enumitem}
\usepackage{datetime}
\usepackage{epstopdf}
\usepackage{xcolor}
\usepackage{lipsum}
\usepackage{multirow}
\usepackage[makeroom]{cancel}

\makeatletter
\DeclareFontFamily{U}{tipa}{}
\DeclareFontShape{U}{tipa}{m}{n}{<->tipa10}{}
\newcommand{\arc@char}{{\usefont{U}{tipa}{m}{n}\symbol{62}}}%

\newcommand{\arc}[1]{\mathpalette\arc@arc{#1}}

\newcommand{\arc@arc}[2]{%
	\sbox0{$\m@th#1#2$}%
	\vbox{
		\hbox{\resizebox{\wd0}{\height}{\arc@char}}
		\nointerlineskip
		\box0
	}%
}
\makeatother

\newtheorem{theorem}{Theorem}
\newtheorem{corollary}{Corollary}

\newtheorem{definition}{Definition}

\newtheorem{proposition}{Proposition}
\newtheorem{remark}{Remark}

\DeclareMathOperator*{\argmin}{arg\,min}

\DeclareMathOperator{\cL}{\mathcal{L}}

\DeclareMathOperator{\cR}{\mathcal{R}}

\DeclareMathOperator{\SINR}{\mathrm{SINR}}

\DeclareMathOperator{\bR}{\mathbb{R}}

\DeclareMathOperator{\bP}{\mathbf{P}}
\DeclareMathOperator{\bS}{\mathbb{S}}
\DeclareMathOperator{\ind}{\mathbbm{1}}
\DeclareMathOperator{\bE}{\mathbf{E}}

\newcommand*\diff{\mathop{}\!\mathrm{d}}

\newcommand*\nnb{\nonumber}

\definecolor{sandy}{HTML}{E6E2AF}
\definecolor{stone}{HTML}{A7A37E}
\definecolor{beach}{HTML}{EFECCA}
\definecolor{ocean}{HTML}{046380}
\definecolor{diver}{HTML}{002F2F}

\definecolor{Firenze1}{HTML}{468966}
\definecolor{Firenze2}{HTML}{FFF0A5}
\definecolor{Firenze3}{HTML}{FFB03B}
\definecolor{Firenze4}{HTML}{B64926}
\definecolor{Firenze5}{HTML}{8E2800}
\definecolor{mediumpersianblue}{rgb}{0.0, 0.4, 0.65}
\definecolor{hongik}{HTML}{004498}
\definecolor{cobalt}{rgb}{0.0, 0.28, 0.67}
\definecolor{burntorange}{rgb}{0.8, 0.33, 0.0}

\definecolor{ultramarineblue}{rgb}{0.25, 0.4, 0.96}

\title{Dynamical System Characterization of Heterogeneous Walker Satellite Networks: \\An Orbit-Aware Stochastic Geometry Perspective}

\author{Chang-Sik Choi
  	 and Francois Baccelli
	\IEEEcompsocitemizethanks{\IEEEcompsocthanksitem{Chang-Sik Choi is an Assistant Professor at School of Electrical Engineering, KAIST, South Korea.  (email: changsik@kaist.ac.kr).
			Francois Baccelli is with Telecom Paris and Inria Paris (email: francois.baccelli@ens.fr).
}}}
\begin{document}
	\maketitle

\begin{abstract}
Heterogeneous and in particular multi-altitude low Earth orbit (LEO) satellite constellations exhibit complex spatial and temporal structures, which require new modeling tools for their performance analysis. In this paper, we develop an orbit-aware stochastic geometry framework modeling today's LEO satellites on various orbits and various altitudes. In particular, we characterize such a system as the superposition of multiple Walker point processes and formulate it as a dynamical system determined by an initial condition and the rotation speeds of satellites and Earth. We show that when the speeds are rationally commensurable, the proposed satellite system is periodic. Then, we show that the system is ergodic when the speeds are rationally independent, establishing a theoretical link between time averages of the system and the expectation of it under the invariant measure. We derive the nearest-satellite distance distribution of a typical receiver at a given latitude and analyze the signal to interference-plus-noise ratio (SINR) coverage probability of the typical receiver. We then derive the ergodic throughput of the downlink communication to the typical receiver. Overall, the proposed framework offers a rigorous and tractable tool for analyzing downlink performance in Walker-type heterogeneous LEO satellite networks.

\end{abstract}


\section{Introduction}
\subsection{Motivation and Related Work}
\IEEEPARstart{S}{atellites} in Earth orbit support diverse applications, from communication to sensing, data harvesting, and Earth observation~\cite{8700141,9210567}. In particular, low Earth orbit (LEO) satellites have drawn significant attention for communication purposes due to their short propagation distance and low latency. As the cost of space launch decreases, the number of satellites is expected to grow, making their spatial distribution very complex. A widely used basis for structured satellite deployment is the Walker constellation~\cite{walker1984satellite}, in which satellites are periodically-spaced on periodically-spaced orbital planes. Nevertheless, practical satellite distributions become increasingly complex to be captured merely by a single Walker model, since most constellations are multi-altitude multi-inclination and also since satellites from multiple operators coexist across different orbital planes and altitudes~\cite{10436074,8626457,9502642,FCCSTARLINK,FCCKuiper,FCCBoeing}. Since the spatial distribution of satellites directly affects the performance of communications, constructing a tractable model for today’s complex satellite constellations is a key first step toward tractable analysis and optimization of current and forthcoming satellite networks.

To model such satellite distributions, many approaches have been proposed and studied. In the context of stochastic geometry~\cite{baccelli2010stochastic,haenggi2012stochastic,chiu2013stochastic}, a spherical binomial point process was initially used to model the locations of satellites~\cite{9079921,9177073,9218989,9497773,9678973}. The binomial point process features points randomly distributed on a spherical surface. Using this model, \cite{9079921,9177073,9218989,9497773,9678973} analyzed key communication performance metrics. Extending the spherical binomial model, a spherical Poisson point process was employed to describe satellite locations as a random number of uniformly distributed points on a sphere. The difference between the binomial and Poisson models lies in the number of satellites, while the interactions between points are identical for both models: specifically, satellites are uniformly and independently distributed on a spherical surface, hence forbidding any clustering or repulsion. These previous studies also showed that the uniformly distributed points can be used to approximate the performance of several existing satellite networks. Nevertheless, the binomial or Poisson model has no orbital structure for satellites, so that the motion of satellites cannot be examined under these models. Furthermore, in commercial LEO constellation deployments~\cite{FCCSTARLINK}, where many satellites share each common orbits, there is strong spatial clustering between satellites, and, as widely reported in~\cite{5208529,8340239,8357962}, the clustering of transceivers leads to communication performance very different from that obtained from network layouts without spatial clustering. Therefore, it is important to describe the orbital structure of satellite networks, not only to capture the clustering of satellites, but also to establish a unified framework to examine the motions of satellites.

Recently, a spherical Cox model was developed in~\cite{10410220,10557592,10703111} to account for the orbits of LEO or MEO satellites. In this model, satellites are generated on a geometric orbital structure so that the satellite points are always located on the orbits. Several studies used this model to analyze downlink communications from heterogeneous satellites \cite{10410220,10557592,10703111}, coverage behavior combined with aerial platforms or with satellite cells~\cite{10771991,Choi_Coveragecell}, and delay-tolerant data harvesting~\cite{10436110} based on LEO or MEO satellites. In addition to the fact that the Cox model characterizes the orbital structure of satellite networks, the Cox model also allows one to assess the dynamical aspects of the LEO satellite network architecture, which had been neglected due to the absence of the appropriate tools. A temporal analysis of the Cox model in stochastic geometry was conducted in~\cite{10436110}, where both snapshot-based performance and time-domain behavior were examined. Although the Cox model introduces the orbital structure and locally approximates forthcoming constellations, the orbit distribution used in these studies is assumed to be isotropic for analytical tractability. Consequently, the Cox model remains unable to capture the periodic structure and non-isotropic characteristics of orbital layouts in today’s complex satellite deployments.

In a related context, to describe non-isotropic orbital distribution, \cite{11159552} characterized the spatial distribution of a Walker constellation over time as a stochastic geometry model with a dynamical system structure. Specifically, \cite{11159552} identified the conditions under which the single-inclination and single-altitude (SISA) Walker point process is periodic or not, thereby providing the mathematical justification for spatial and temporal analysis. \cite{11159552} considered a simple Walker constellation where all orbits have the same inclination and the same radius, and satellite phase differences are all identical.

Nevertheless, practical constellations are not composed of such an ideal Walker constellation with SISA orbits. Most existing or forthcoming real-world deployments feature heterogeneous constellations where multiple orbital structures are combined across several layers\cite{FCCSTARLINK,FCCKuiper,FCCBoeing}. For instance, in Starlink’s second generation deployment plan~\cite{FCCSTARLINK}, the satellite constellation is heterogeneous and divided into three independent Walker components, each characterized by different inclination angles, satellite altitudes, orbit counts, and satellite counts. To characterize the heterogeneous orbital foundations and their satellites, we develop a new model that reproduces the heterogeneous and in particular multi-altitude nature of current and future constellations within the framework of stochastic geometry. Our new orbital structure enables the analysis of both the time-domain and the spatial-domain performance of downlink communication in a unified framework.

\subsection{Theoretical Contributions}
The theoretical contributions of this work are three-fold.

\subsubsection{A Tractable Model for Heterogeneous Multi-altitude Satellite Networks}
We develop a stochastic geometry framework that models the spatial layout of modern heterogeneous satellite networks with periodically distributed satellites on orbits of various altitudes. Specifically, our model is based on an independent superposition of Walker point processes, where each Walker point process features evenly spaced SISA orbits and evenly spaced satellites conditionally on the orbits. The proposed point process model is formulated as a dynamical system, parameterized by its initial condition and by speed constants which capture the rotation speeds of Earth and satellites. 

\par By combining multiple Walker point processes in a single unified framework, we present an analytical approach to fully characterize both the spatial configuration and its temporal evolution. In contrast to~\cite{11159552}, our model allows satellite inclinations and altitudes to take multiple values, thereby precisely capturing the heterogeneous structure of real-world practical satellite networks alluded to above. In the proposed framework, the speeds of satellites on different orbital layers can be handled as individual parameters, so that the complex motion of multi-altitude heterogeneous satellite networks is described naturally within our unified framework.

\subsubsection{Time Domain Analysis based on Dynamical System Theory}
Under our unified space–time framework, we analyze the temporal behavior of heterogeneous satellite networks constructed from multiple Walker bases. We derive the exact conditions under which the resulting dynamical system is periodic, showing that it depends on the Earth’s rotation speed and the speeds of satellites at different altitudes. More precisely, if those speeds are rationally commensurable, the system is periodic; on the other hand, if the speeds are rationally independent, the system is minimally ergodic (see below). Furthermore, in this case, we determine the dynamical system is ergodic with respect to (w.r.t.) this invariant measure~\cite{Katok}. The implications are significant: based on Birkhoff’s theorem, in the ergodic scenario, long-term time averages coincide with the expectations under the invariant distribution of the network architecture. This leads to a formal mathematical justification for the long-term time-average analyses that have been commonly used and adopted in network analysis.

\subsubsection{Performance Analysis of Downlink Communications}
Using the exact conditions derived for our dynamical system, we conduct a system-level analysis of downlink communications from satellites to receivers on Earth. Under the ergodic condition, we analyze the distance from the typical receiver to its nearest satellite. We obtain this nearest-distance distribution as a function of key parameters. In particular, since the satellite distribution is non-isotropic, we show that receivers at different latitudes observe different satellite distributions and thus different nearest-distance distributions. In addition, under ergodicity, we show that the derived nearest distance distribution is identical to the time average of the nearest distance seen from the typical receiver, measured over a very long time horizon. Under ergodicity, we analyze the signal to interference-plus-noise ratio (SINR) coverage probability of the typical receiver under various fading scenarios and obtain the coverage formula as a function of key system parameters, then connect the ergodic SINR coverage probability of the typical receiver to the time average of SINR coverage indicators of receivers on Earth. We analyze the large-scale behavior of the coverage probability w.r.t satellite altitudes and orbit inclinations. Finally, we derive the ergodic rate of the proposed network architecture. These performance metrics are crucial for building and optimizing communication networks based on LEO or MEO satellites, as well as any architecture influenced by interference from such systems.

\section{System Model}
\subsection{Spatial Model}\label{S:II-A}

In this paper, the origin of the three-dimensional Euclidean space is at the center of the Earth. The reference space rotates with the Earth’s spin so that a fixed location on Earth has constant $(x,y,z)-$coordinates. The $xy$-plane is the equatorial plane, the $x$-axis is the longitude reference axis (East), and the $z$-axis is the North pole axis. The receivers are randomly distributed at fixed locations on the surface of Earth of radius $r_e = 6371$ km. 

This paper assumes that all satellite orbits are circular. The orbits do not rotate along with the Earth and the satellites on the orbits rotate around the Earth. The rotating speeds of satellites on their orbits are determined by orbit altitudes.

We adopt the standard notation that the orbit longitude is defined as the angle from the $x$-axis to the ascending point where the ascending point of an orbit is the point at which the orbit intersects the reference plane\footnote{The orbit longitude in this paper corresponds to RAAN(Right Ascension of Ascending Node)}. The inclination of an orbit is defined as the angle that the orbital plane makes with the equatorial plane. Finally, the phase of a satellite is defined as the angle from the ascending point to the satellite at that time, measured in the orbital plane of the satellite.

In this paper, we model the satellite constellation as an independent collection of SISA Walker point processes, denoted by $\Psi = \sum_{i=1}^{I}\Psi_i$, where $\Psi_i$ denotes a SISA Walker point process. Each SISA satellite points $\Psi_i$ is on $S_{r_i}$, the sphere of radius $r_i > r_e$. For a given $\Psi_i$, let $N_i$ be the number of orbits and $M_i$ be the number of satellites on each orbital plane. The inclination angle of all orbits is denoted by $\phi_i$. For each SISA Walker point process $\Psi_i,$ the longitudes of the orbits are evenly spaced over $ 2\pi$ and the phases of the satellites on each orbit are evenly spaced over $2\pi$.

We assume that the longitude of the $1$-st orbit of the $i$-th SISA $\{\theta_{i,1}\}_{i=2,...,I}$ is given as a deterministic function of $\theta_{1,1}$, the longitude of the $1$-st orbit of $1$-st SISA Walker point process: namely, $\{\theta_{i,1}=f_i(\theta_{1,1}) \mod 2\pi/N_i\}_{i=2,...,I}$, where $f_i$ is some deterministic function. This means that the longitude of any orbit of $\Psi$ is represented as a function of $\theta_{1,1}$ the longitude of the $1$-st orbit of the SISA Walker point process $\Psi_1.$

Let $\theta_{1,1}$ be the initial longitude of the $1$-st orbit of $1$-st SISA Walker point process at time zero. Then, for $i=1,...,I,$ the initial longitude of the $j$-th orbit of the $i$-th SISA Walker point process is given by
\begin{align}
	\theta_{i,j} &= \theta_{i,1}+\frac{2\pi (j-1)}{N_{i}} \ j=2,...,N_i \label{eq:long}\\
				 &= f_i(\theta_{1,1}) +\frac{2\pi (j-1)}{N_i}  \ j=2,...,N_i.\label{eq:long2}
\end{align}
where $N_i$ is the number of orbits of the $i$-th SISA Walker point process. 
 
Similarly, let $\omega_{i,j,1}$ be the initial argument angle of the $1$-st satellite of the $1$-st orbit of the $i$-th SISA Walker point process where $i=1,...,I$ at time zero. For simplicity, we also write ${\omega}_i \equiv  \omega_{i,j,1}$. Then, the argument angle of the $k$-th satellite on the $j$-th orbit of the $i$-th SISA Walker point process is 
\begin{align}
	\omega_{i,j,k} &= {\omega}_{i,j,1} + \frac{2\pi {(k-1)}}{M_i},   k=2,...,M_i,\label{eq:13}\\
	&= {\omega}_{i} + \frac{2\pi {(k-1)}}{M_i},   k=2,...,M_i,\label{eq:3}
\end{align}
where $M_i$ is the number of satellites per orbit of the $i$-th SISA Walker point process. 

\begin{figure}
	\centering
	\includegraphics[width=1\linewidth]{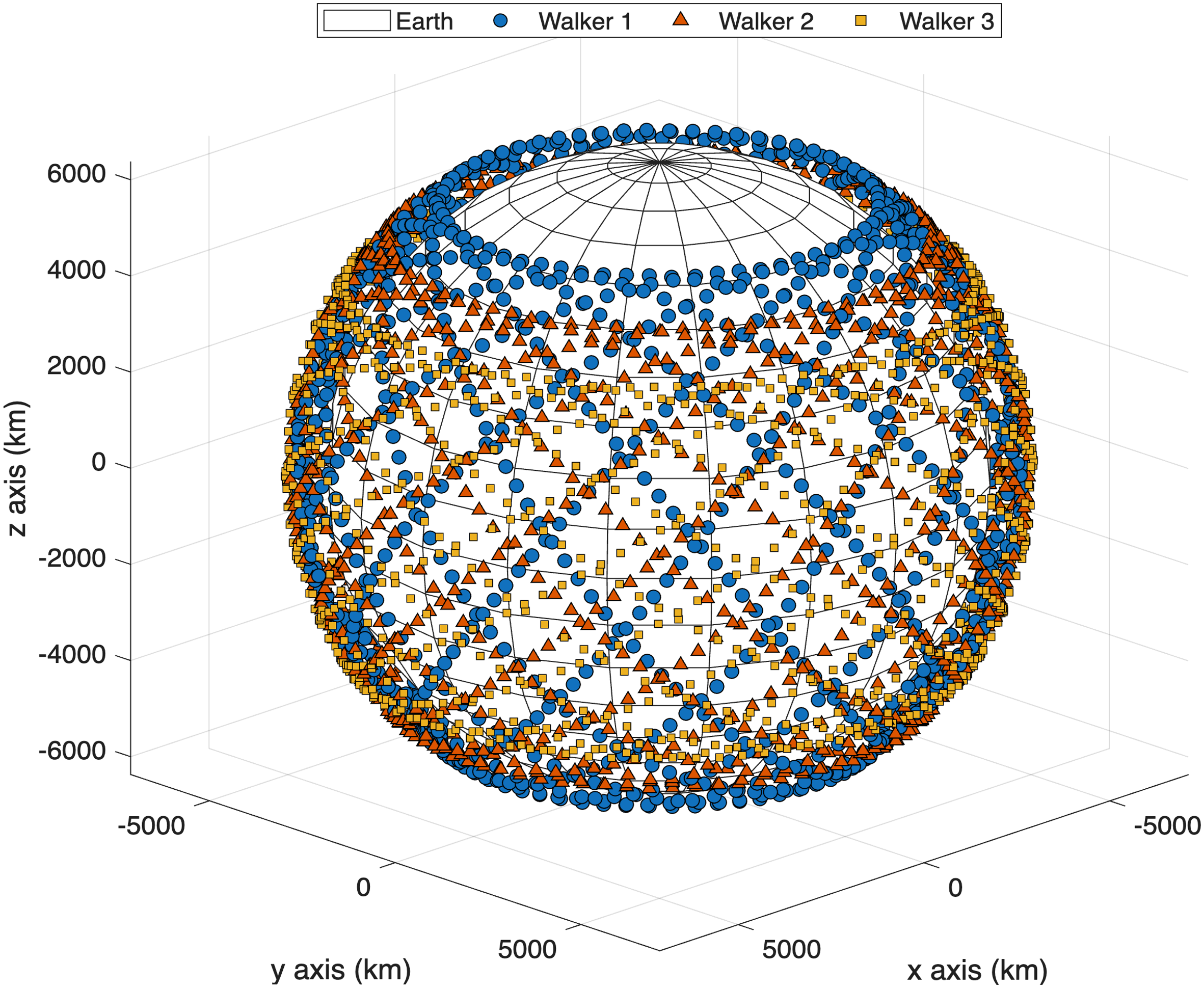}
	\caption{The proposed satellite network model based on $3$ SISA Walker point processes.}
	\label{fig:figure1}
\end{figure}

Fig.~\ref{fig:figure1} illustrates the proposed network model with $3$ SISA Walker point processes: $\Psi=\Psi_1+\Psi_2+\Psi_3$, where for $\Psi_1$,  $N_{1}=15$, $M_{1}=100$, $\phi_{1}=53^{\circ}$, and $r_1=6921$ km. For $\Psi_2,$ $N_{2}=15$, $M_{1}=100$, $\phi_{1}=43^{\circ}$, and $r_2=7021$ km. Finally, for $\Psi_3, $ $N_{3}=15$, $M_{3}=100$, $\phi_{3}=33^{\circ}$, and $r_3=7121$ km. In this figure, we use $\theta_{2,1}=\theta_{1,1}+({2\pi}/{N_1})({1}/{3})$ and $\theta_{3,1}=\theta_{1,1}+({2\pi}/{N_1})({2}/{3}).$

Let $X_{i,j,k}$ denote the coordinates of the $k$-th satellite on the $j$-th orbit of the $i$-th SISA Walker point process.
Let $\omega_{i,j,k}$ denote the phase of this satellite and let $\theta_{i,j}$ denote the longitude of the
$j$-th orbit of the $i$-th SISA Walker point process. Then, the 
Cartesian coordinates of the initial location of the satellite $X_{i,j,k}$ are given by
\begin{align}
	x&=r_{i}\rho_{i,j}\cos(\theta_{i,j}+\breve{\theta}_{i,j,k}),\label{xyz}\\
	y&=r_{i}\rho_{i,j}\sin(\theta_{i,j}+\breve{\theta}_{i,j,k}),\label{xyz2}\\
	z&=r_{i}\sin(\omega_{i,j,k})\sin(\phi_{i}),\label{xyz3}
\end{align}
where we use the notation 
\begin{eqnarray*}
	\rho_{i,j} & = & \sqrt{\cos^2(\omega_{i,j,k})+\sin^2(\omega_{i,j,k})\cos^2(\phi_{i})},  \\
	\breve{\theta}_{i,j,k} & = & \mathrm{atan2}(\sin(\omega_{i,j,k})\cos(\phi_{i}),\cos(\omega_{i,j,k})). 
\end{eqnarray*}
In a vector notation, we can also write 
\begin{multline}
	\vec{X}_{i,j,k}= (r_{i}\rho_{i,j}\cos(\theta_{i,j}+\breve{\theta}_{i,j,k}),\\r_{i}\rho_{i,j}\sin(\theta_{i,j}+\breve{\theta}_{i,j,k}), r_{i}\sin(\omega_{i,j,k})\sin(\phi_{i}))
\end{multline}

Finally, the point process $\Psi_i$ of the $i$-th SISA constellation is defined as
\begin{equation}
	\label{eq:fin}
	\Psi_i=\sum_{j=1}^{N_i} \sum_{k=1}^{M_i} \delta_{X_{i,j,k}}
\end{equation}
and the whole constellation point process as $\Psi=\sum_{i=1}^I \Psi_i$.

\subsection{Satellite Dynamic}\label{s:dynamic}
In the proposed system, there are two types of motions. The first one is the motion of satellites along
their orbits and the second one is the motion (spin) of Earth.
First, the law of gravity dictates the speed of satellites and for a given Walker point process $\Psi_{i}$ 
of radius $r_i$, the speed of the satellites is denote by $v_{i}$. Let $v_e$ be the speed of Earth rotation.
Here, $v_i$ and $v_e$ are angular speeds $( \text{rad}/\text{sec}).$ Then, for the $i$-th SISA Walker point process,
the longitude of the $j$-th orbit and the phase of the $k$-th satellite on the $j$-th orbit at time $t$ are given by
\begin{align}
	\theta_{i,1}(t) &= \left.\theta_{i,1}-v_e t \mod 2\pi/N_i\right.,\\
	\omega_{i,j,1}(t) &= \left.\omega_{i,j,1}+v_it \mod 2\pi/M_i \right.,
\end{align}
respectively, where $\theta_{i,1}$ is the initial longitude of the $1$-st orbit of the $i$-th SISA Walker
point process, and $\omega_{i,j,k}$ is the initial phase of the $k$-th satellite on the $j$-th orbit of 
the $i$-th SISA Walker point process as defined above.

Then, using the relation in Eq. \eqref{eq:long}, the longitude of the $j$-th orbit of the $i$-th constellation at time $t$ is
\begin{align}
	\label{eq:10}
	&\theta_{i,j}(t) = \theta_{i,1}(t) + \frac{2\pi (j-1)}{N_{i}}, \  j=2,...,N_i,
\end{align}
for $i=1,...,N_i.$
 
Using Eq. \eqref{eq:3}, the phase angle of the $k$-th satellite on the  $j$-th orbit of the $i$-th constellation at time $t$ is
\begin{align}
	\label{eq:11}
	\omega_{i,j,k}(t) &= \omega_{i}(t) + \left.\frac{2\pi(k-1)}{M_i} \right., \  k=2,...,M_i.
\end{align}
for $i=1,...,N_i$.

The equations \eqref{eq:long}--\eqref{eq:fin} applied above
to the initial condition can be applied to the last angles 
to define the positions $X_{i,j,k}(t)$, $i=1\ldots,I$, $j=1,\ldots,N_i$ and $k=1,\ldots,M_i$ of
the satellites of the constellation at time $t$ and
then the associated point processes $\Psi_i(t),\ i=1,\ldots,I$ and $\Psi(t)=\sum_{i=1}^I \Psi_i(t)$.

\subsection{Association and Propagation Model}
We assume that each receiver is associated with its closest satellite transmitter. 
Let $d$ be the distance from a satellite transmitter to its associated receiver on Earth.
At any given time, the received signal power is modeled as $ p G(\eta(d)) H d^{-\alpha}$, 
where $p$ denotes the $1$-meter reference average received signal power,
$G(\eta(d))$ denotes the antenna gain over a distance $d$,
$H$ denotes the random variable for general small-scale fading independent of everything else\footnote{Specifically, $H$ is a sample from a stationary and mixing fading process independent of everything else.}, 
and $\alpha$ denotes the path-loss exponent, greater than or equal to $2$.
We denote by $\bar{F}_H(x)$ the complementary cumulative distribution
function (CCDF) of the random variable $H$ evaluated at $x$.

Here, the antenna gain $G(\eta(d))$ is a function of $\eta(d)$ the angle from 
the transmitter’s boresight direction to its receiver at distance $d$ from the transmitter.
When the satellite transmitter’s orbit radius is $r_i$, we use trigonometry
to obtain the relation $\eta(d) = \arccos\!\left({d^2 + r_i^2 - r_e^2}/{(2 d r_i)}\right)$.

Motivated by practical scenarios discussed in \cite{38821}, we assume a simple model
of the form $G(\eta(d)) = g_m$ if $\eta(d) < \eta_c$ and $G(\eta(d)) = 1$, otherwise.
For downlink communications, we assume that the serving transmitter can implement
the maximum gain toward its dedicated receiver, denoted by $g_M$, where we assume $g_M > g_m \geq  1$.

\subsection{Performance Metric}
In the first part, we characterize the proposed model as a dynamical system with its initial
condition and the dynamics determined by the relative motion of Earth and that of the satellites on orbits. 
We provide the invariant measure of the dynamical system and derive the exact conditions under which the
dynamical system is periodic or not and is ergodic w.r.t. this invariant measure or not.

In the second part, under the condition that the dynamical system is ergodic, we analyze the
performance of downlink communications from satellites to receivers on Earth. In particular,
we analyze the distance from the typical receiver to its association satellite, the SINR
coverage probability of the typical receiver, and the Laplace transforms of the interference
and the total received power. These performance metrics are obtained as functions of key
network parameters. By establishing the ergodicity condition of the proposed dynamical system
w.r.t. the invariant measure and then deriving the network performance using it, we not 
only analyze the performance of the proposed network seen from the typical receiver but
also derive the time averages of any given receiver over a long time. 

\section{Dynamical System Analysis}
Let $\mathbb S $ be the $(I+1)$-dimensional torus defined by the product space formed by the
orbit longitude $\theta$ and the satellite phases $\{\omega_{i}\}_{i=1,...,I}$ as follows:
\begin{equation} 
	\mathbb S = \left[0,\frac{2\pi}{N_1}\right) \times \prod_{i=1}^I\left. \left[0,\frac{2\pi}{M_i}\right)\right.\label{eq:S}.
\end{equation}

Based on the motions of both Earth and satellites, the dynamic of the satellites 
is characterized by the continuous flow $\{R_t(x)\}_t$ on $\mathbb S$. Here $R_t(x)$ denotes
the state at time $t$ for the initial condition 
$x=\{\hat{\theta}_{1,1},\{\hat{\omega}_{i,j,k}\}_{i=1,...,I}\}\in \mathbb S$, that is
\begin{multline}
	R_t(x)
	=\left(\hat{{\theta}}_{1,1}-v_et \right) \mod{\frac{2\pi}{N_i}},\\
	 \left\{\left(\hat{\omega}_{i,j,1}+v_it \right)\mod{\frac{2\pi}{M_i}}\right\}_{i=1,...,I},
\end{multline}
for $t>0$. As above, $v_e$ is the Earth rotation angular speed and $v_i $ is the satellite rotation
angular speed of $\Psi_i$, the $i$-th SISA Walker point process.
\begin{definition}
The speed constants $v_e,v_1,\dots,v_I$ are said to be rationally independent if $a_0v_e+a_1v_1+\cdots+a_{I}v_I=0$
 for the integers $a_0,a_1,...a_I\in \mathbb Z$, implies $a_i=0$ for all $i$. Otherwise, the speed constants are said to be
rationally dependent. As a special case of rational dependence, the speed constants are said to be rationally
commensurable if they are all rational multiples of a common constant.
\end{definition}
Note that for $I=1$, rational dependence is equivalent to rational commensurability. 
However for $I>1$, rational commensurability is a strict subclass of rational dependence.
\begin{definition}
A dynamical system $\alpha_t(x)$ on a compact state space $\mathbb X$ with initial condition $x$ is minimal
if every {orbit} is dense in $X$, namely, for all $ x\in \mathbb X$
	\begin{equation}
		\overline{\{\alpha_t(x):t\in\bR^+\}}=\mathbb X,
	\end{equation} 
where $\overline{\{\alpha_t(x), t\in\bR^+\}}$ denotes the closure of the orbit with the initial condition $x\in \mathbb X$
\cite{katok1995introduction}.
\end{definition}

\begin{proposition}\label{Prop:1}
If the speeds $\{v_e,v_1,v_2,...,v_I\} $ are rationally independent, the dynamical system $\{R_t\}$ 
is aperiodic and minimal on the torus $\mathbb S$.
If the speeds are rationally commensurable, for all initial conditions, the dynamical system is periodic.
If the speeds are rationally dependent but not rationally commensurable,
the dynamical system is aperiodic and not minimal. 
\end{proposition}

\begin{IEEEproof}
The satellite dynamic is defined on the $(I+1)$-dimensional torus $\mathbb S$ where $I$ dimensions
come for satellite rotation of each of the $I$ SISA Walker point processes and the extra dimension
comes for the Earth rotation. It is well known that the linear flow on the torus is periodic if
the speed constants are rationally commensurable, whereas it is minimal 
if the speed constants are rationally independent \cite{Katok}.
\end{IEEEproof}

In the rest of this paper, we mostly focus on the following two cases: rationally independent and rationally
commensurable speeds. The third case is interesting too but skipped for the safe of compactness.

Let $\overline{\mathcal Q}$ be 
the product of the uniform distributions on the torus $\mathbb S$ defined as follows:
\begin{equation}\label{eq:Q}
\overline{\mathcal Q}\sim \mathrm{Uniform}\left(0,\frac{2\pi}{N_1}\right)\times \prod_{i=1}^I
\mathrm{Uniform}\left(0,\frac{2\pi}{M_i}\right).
\end{equation}

\begin{theorem}\label{Theorem:1:basic}
Regardless of the relations on the speeds, $\overline{\mathcal{Q}}$ is an 
invariant measure for the dynamical system $\{R_t\}$ on $\mathbb S$.
\end{theorem}
\begin{IEEEproof}
$\overline{\mathcal{Q}}$ is the Haar measure on $\mathbb S$
and thus it is an invariant measure for this dynamical system \cite{walters2000introduction,Katok}.
\end{IEEEproof}

\begin{theorem}\label{Theorem:speedratio}
If the speeds $\{v_e,v_1,v_2,...,v_I\} $ are rationally independent,
the dynamical system $\{R_t\}$ is ergodic w.r.t. the invariant measure $\overline{\mathcal Q}$.
More precisely, $\{R_t\}$ is uniquely ergodic.
\end{theorem}

\begin{IEEEproof}
The proof is similar to that of the two-dimensional torus proved in
\cite{11159552} where we employ the classical result on the discrete dynamic on two-dimensional torus \cite{Katok}. 
\end{IEEEproof}

Ergodicity implies that for all non-negative measurable functions $h$,
and for all $x$ in a subset of $\mathbb S$ of $\overline{\mathcal Q}$ measure 1, we have 
	\begin{equation}
		\lim_{T\to \infty} \frac 1 T \int_0^T h(R_t(x)) \mathrm{d} t =
		\int h(r) {\overline{\mathcal Q}}(\mathrm{d} r) .
	\end{equation}
That is, long term time averages match space averages (i.e., averages w.r.t. $\overline{\mathcal Q}$) for
$\overline{\mathcal Q}$-almost all initial conditions.
Unique ergodicity \cite{walters2000introduction} is stronger. 
It implies that, for all initial conditions $x\in\mathbb S$,
and for all continuous non-negative functions $f$ on $\mathbb S$, 
	\begin{equation}
		\lim_{T\to \infty} \frac 1 T \int_0^T f(R_t(x)) \mathrm{d} t =
		\int f(r) {\overline{\mathcal Q}}(\mathrm{d} r) .
	\end{equation}
The key new thing is the extension from ``almost all" to "all" initial conditions when restricting to
continuous functions.

Above, we built the flow $\{R_t\}$ on the probability space $(\mathbb{S},\mathcal{S}, \overline{\mathcal Q})$
where $\mathbb S$ is a complete and separable metric space (which is also compact),
$\mathcal S$ is the Borel sigma-field of $\mathbb S$, and $\overline{\mathcal Q}$ is as defined above.

Consider now the family of point processes in the space of counting measures with
a constant number $K_i=N_iM_i$ of points on each of the spheres of radii $r_i$, $i=1,\ldots,I$,
with $N_i$ equally spaced orbits, all with the same inclination $\phi_i$, and such that all orbits
have $M_i$ equally spaced satellites and the same phase for the first satellite on all $N_i$ orbits.
Note that the space $\mathbb W$ of Walker point processes with these characteristics
space is also a compact, complete and separable metric space.
Let $\mathcal{W} $ denote the Borel $\sigma$-algebra on $\mathbb W$.

Let $\Xi:\mathbb{S} \to \mathbb{W}$ be the map defined by Eqs. 
\eqref{eq:long}--\eqref{eq:fin}, which associates the point process $\Psi$ to
the angles $\theta$, $\{\omega_i\}_{i=1,\ldots,I}$, where $\theta$ is the longitude of the $1$-st
orbit of the $1$-st SISA Walker constellation and $\omega_i$ is the phase of the $1$-st
satellite of the $1$-st orbit of the $i$-th constellation where $i=1,\ldots,I$. 

Let $\overline{\mathcal{Q}}_{\Psi}$ be the pushforward on $(\mathbb W,\mathcal{W})$ of
$\overline{\mathcal{Q}}$ by the map $\Xi$. Similarly, let $\{\Psi_{t}\}$ be the satellite point process at time $t$, which is defined as follows: 
\begin{equation}
	\Psi_t = \Xi(R_t), \quad t\ge 0.
\end{equation}
If $R_0=x$, then $\Psi_0=\Xi(x)$, namely $\Xi(R_t(x))=\Psi_t(\Xi(x))$ for all $x$ and all $t$. 

\begin{theorem}\label{T:rev1:conjugatcy}
	The dynamical systems 
	$\{R_t\}$ on $(\mathbb S,\mathcal{S},\overline{\mathcal Q})$ and
	$\{\Psi_t\}$ on $(\mathbb{W},\mathcal{W}, \overline{\mathcal{Q}})$ are isomorphic.
\end{theorem}

\begin{IEEEproof}
The map $\Xi$ is continuous from the metric space $\mathbb{S}$ to the metric space $\mathbb{W}$.
It is invertible and its inverse $\Xi^{-1}$ is also continuous. 
This map preserves the measures and the flows as well.
Consequently, these two measure-preserving dynamical systems are isomorphic.  
\end{IEEEproof}

\begin{corollary}\label{T:rev2:conjugatcy}
For all speeds, $\overline{\mathcal{Q}}_{\Psi}$ is an invariant measure for  $\{\Psi_t\}$.
If the speeds $\{v_e,v_1,v_2,...,v_I\}$ are rationally independent, the dynamical system  
$\{\Psi_t\}$ is minimally and uniquely ergodic on $(\mathbb{W},\mathcal{W})$ w.r.t.
the invariant measure $\overline{\mathcal{Q}}_\Psi$. 
If the speeds $\{v_e,v_1,v_2,...,v_I\} $ are rationally commensurable, the dynamical system $\{\Psi_t\}$
is periodic for all initial conditions.
\end{corollary}
\begin{IEEEproof}
The result follows from the fact that the two systems are isomorphic.
Then, we use Proposition \ref{Prop:1} and Theorem \ref{Theorem:speedratio} to complete the proof. 
\end{IEEEproof}

\begin{remark}
If the speeds $\{v_e,v_{1},...,v_{I}\} $ are rationally independent, the system is minimally and
uniquely ergodic w.r.t $\overline{\mathcal{Q}}_{\Psi}$, which implies that,
for all continuous functions $f:\mathbb W \to \mathbb R^+$, for all initial conditions $w\in \mathbb W$,
\begin{equation}
\lim_{T\to \infty} \frac 1 T \int_0^T f(\Psi_t(w)) \mathrm{d} t =
\int f(\psi) {\overline{\mathcal Q}}_{\Psi}(\mathrm{d} \psi)  ,
\label{eq:ergodic}
\end{equation}
where the left-hand side is the time average of the function $f$ for the initial condition $w$ 
and the right-hand side is the mean of the function $f$ w.r.t. the invariant measure ${\overline{\mathcal Q}}_{\Psi}$.
\end{remark}

\begin{remark}
Consider a fixed location $U=(x,y,z)$ in the referential adopted in this paper, for instance a static location on Earth (hence rotating with it). Let us define the flow 
$$\Psi^U_t = \sum_{i,j,k} \delta_{X_{i,k,k}(t) - U},$$ 
namely the constellation ephemeris point process seen by a static observer at $U$, evaluated in
this basic referential with the origin shifted to $U$.
The space in which $\Psi^U$ lives is again a metric space $\mathbb{W}^U$ of counting measures.
There is again a continuous bijection $\chi$ from  $\mathbb{S}$ to $\mathbb{W}^U$ 
such that $\chi^{-1}$ is continuous.
Let $\overline{Q}_{\Psi^U}$ be the pushforward of $\overline{Q}$ by the map $\chi$ on the
Borel sigma-field of $\mathbb{W}^U$. 

Then, by construction, $\Psi_t^U=\chi(R_t(x))$ for all $t$ and $x\in \mathbb S$. Therefore,
the dynamical systems $\{R_t\}$, and $\{\Psi_t^U\}$ are isomorphic \cite{walters2000introduction}.
In all cases, $\overline{Q}_{\Psi^U}$ is an invariant measure for $\{\Psi^U_t\}_t$.
	
If the speeds are rationally commensurable, $\{\Psi^U_t\}$,
the satellite ephemeris seen from the fixed observer, is periodic. In this case, this dynamical system admits
another invariant measure for each initial condition, which is an averaging over the period \cite{katok1995introduction}.

If the speeds are rationally independent, the dynamical system $\{\Psi_t^U\}$
is minimal and uniquely ergodic w.r.t. its unique invariant measure $\overline{Q}_{\Psi^U}$.
\end{remark}

\section{Performance Analysis: Ergodic Scenario}
Here, we investigate the downlink communication performance of the proposed system when it is ergodic, namely, the speeds are rationally independent. Specifically, we derive the performance of communications from satellites to the typical receiver at
\begin{equation}
	\vec{u}=(r_e\cos(l_u),0,r_e\sin(l_u)), 
\end{equation}
where $l_u$ denotes the latitude of the receiver. 

We will derive the distribution of the distance to the nearest satellite (Theorem 4) and the distribution of the SINR of the typical receiver (Theorem 5).

These performance metrics are derived as the spatial average w.r.t. the unique invariant measure $\overline{\mathcal{Q}}$ in Eq. \eqref{eq:Q}. For a simpler notation, let  ${{\theta}}$ denote a uniform random variable between $0$ and $2\pi/N_1$ and ${{\omega}}_i$ denote an independent uniform random variable between $0$ and $2\pi/M_i$ for $i=1,...,I$. Then,  $\overline{\mathcal{Q}}$ is characterized as the distribution of ${{\theta}}\times\prod_{i=1}^I{{\omega}}_{i}$, namely 
\begin{equation}
	\overline{\mathcal{Q}} = \textrm{Dist} \left({{\theta}}\times\prod_{i=1}^I{{\omega}}_{i}\right)
\end{equation}

\subsection{Distance to Nearest Satellite}
Let $S_{r_i}$ be the sphere of radius $r_{i}$ centered at the origin. Then, for $r_{i}-r_e<d<\sqrt{r_{i}^2-r_e^2},$ we define the spherical cap of the typical receiver $u$ associated with the distance $d$ as follows: $$	S_{r_i}(d) = \{(x,y,z)\in S_{r_i}| \|(x,y,z)- \vec{u} \|\leq d \}.$$
When taking $d=\sqrt{r_{i}^2-r_e^2},$ the spherical cap $S_{r_i}(\sqrt{r_i^2-r_e^2})$ is the visible cap seen from the typical receiver, which is also denoted by $\bar{S}_{r_i}.$  Using it, we represent the total visible cap associated with $\Psi$ from the typical receiver as follow: $$\bar{S}=\cup_{i=1}^I\bar{S}_{r_i} = \cup_{i=1}^I S_{r_i}\left(\sqrt{r_{i}^2-r_e^2}\right). $$

For the $i$-th constellation of radius $r_i$, let $D(l_u,r_i)$ be the random variable of the distance from the typical receiver at the latitude $l_u $ to its nearest satellite on the orbit of radius $r_i$. We take $D(l_u,r_i)=\infty$ if $\Psi_i(\bar{S}_{r_i})=\emptyset,$ namely if there is no visible satellite of $\Psi_i$. Let $D(l_u)$ denotes the distance from the typical receiver to its nearest satellite of any orbit. Using the void probability, we derive the following theorem on the distance to the nearest satellite. In below, $\bP $ is w.r.t. the invariant measure ${\overline{Q}}_{\Psi}$.

\begin{theorem}\label{Theorem:2}
	The CCDF of $D(l_u)$, the distance from the typical receiver to its nearest satellite, is given by 
	\begin{equation}
		\bP(D({l_u})>d) = \int_0^{\frac{2\pi}{N_1}}\left(\prod_{i=1}^I \bP(D(l_u,r_i)>d|\theta)\right)\frac{ N_1}{2\pi}\diff \theta\nnb.
	\end{equation}
	Each term $\bP(D(l_u,r_i)>d|\theta)$ is given by the integral 
	\begin{align}
	\int_{0}^{\frac{2\pi}{M_i}}  \!\! \mathbbm{1} \left\{ \frac{r_{i}^2+r_e^2-d^2}{2 r_{i} r_e }> \max_{j,k} \frac{\langle\vec{X}_{i,j,k}, \vec{u}\rangle }{\|\vec{X}_{i,j,k}\|\|\vec{u}\|} \right\} \! \frac{\diff {\omega}}{2\pi/M_i},\nnb
	\end{align}
	where $\langle {\vec{X}_{i,j,k} \cdot \vec{u}} \rangle$ denotes the inner product of the two vectors and $\vec{X}_{i,j,k}$ is a vector defined using Eqs. \eqref{xyz}--\eqref{xyz3}. 
\end{theorem}

\begin{IEEEproof}
	The distance to the nearest satellite is given by the minimum of $\{D(l_u,r_i), i=1,...,I\}$, namely 
	\begin{align}
		D(l_u) = \min\{D(l_u,r_{1}),...,D(l_u,r_I)\}.
	\end{align}
	Since, given ${{\theta}}$, the random variables  $\{D(l_u,r_i)\}_{i=1,...,I}$ are independent, we have
	 \begin{align}
	 	\bP(D(l_u)>d)&=\bE_{}[\bP(D(l_u)>d|{{\theta}})]\nnb\\
	 	&=\bE\left[\prod_{i=1}^I \bP(D(l_u,r_i)>d|{{\theta}})\right]\label{15}.
	 \end{align}
	 Now, conditonially on the random variables ${{\theta}}$ and ${{\omega}}_i$, the CCDF of the distance from typical receiver to the nearest satellite of $\Psi_i$ is given by
	\begin{align}
		\bP(D(l_u,r_i)>d|{{\theta}})
		&=\bE_{{{\omega}}_i}\left[\bE\left[\ind\left\{D(l_u,r_i)>d\right\}|{{\theta}},{{\omega}}_i\right]\right],\nnb
	\end{align}
	where $\ind\{X\}$ denotes the indicator function of the event $X$. 

Let $\vec{X}_{i,j,k} $ be the notation for the vector from the origin to the point $X_{i,j,k} $ whose Cartesian coordinates are given in Eqs. \eqref{xyz},\eqref{xyz2}, and \eqref{xyz3}, respectively.  

Then, $D(l_u,r_i)$ is greater than $d$ if and only if the central angle from the receiver to the rim of the spherical cap $S_{r_i}(d)$ is smaller than central angles from the typical receiver to all points of $\Psi_i.$ By using the Cosine' law, we derive the central angle from the receiver to the rim of the spherical cap $S_{r_i}(d)$ as follows: $\kappa_i(d)=\arccos\left((r_i^2+r_e^2-d^2)/(2r_ir_e)\right)$. Using this, we have
 \begin{align}
	&\ind\left\{D(l_u,r_i)>d\right\} \nnb\\
	&= \mathbbm{1}\left\{\!\kappa_{i}(d) <\min_{(j,k)}\left\{\!\arccos\left(\frac{\vec{X}_{i,j,k} \cdot \vec{u} }{\|\vec{X}_{i,j,k}\|\|\vec{u}\|}\right)\!\right\}\!\right\}\nnb\\
	&= \mathbbm{1}\left\{\!\frac{r_i^2+r_e^2-d^2}{2r_ir_e} >\max_{(j,k)}\left\{\!\left.\frac{\vec{X}_{i,j,k} \cdot \vec{u} }{\|\vec{X}_{i,j,k}\|\|\vec{u}\|}\right.\!\right\}\!\right\}\label{eq:17},
\end{align}
where, in the minima and maxima above, $(j,k)$ ranges over $j=1,\ldots,N_i$ and $k=1,\ldots,M_i$. 
	Here, we express the angle between $\vec{X}_{i,j,k}$ and $\vec{u}$ using the inner product  $\arccos\left((\vec{X}_{i,j,k}\cdot\vec{u})/(\|\vec{X}_{i,j,k}\|\|\vec{u}\|)\right)$. We also use the fact that $\cos(x)$ is decreasing on $[0,\pi]$. Note that the left hand side of Eq. \eqref{eq:17} is a measurable function of $r_i,$ ${{\theta}},$ and ${{\omega}}_{i}.$ Therefore, since  ${{\omega}}_{i}\sim\mathrm{Uniform}(0,2\pi/M_i),$ we have
	\begin{align*}
		&\bE_{\omega_i}[\ind\left\{D(l_u,r_i)>d\right\}|{{\theta}}]\nnb\\
		&=\int_{0}^{\frac{2\pi}{M_i}}  \!\! \mathbbm{1} \left\{ \frac{r_{i}^2+r_e^2-d^2}{2 r_{i} r_e }> \max_{(j,k)} \frac{\vec{X}_{i,j,k} \cdot \vec{u} }{\|\vec{X}_{i,j,k}\|\|\vec{u}\|} \right\} \! \frac{M_i}{2\pi}\diff \omega.
	\end{align*}
	Finally, we obtain the final result by plugging the last expression to Eq. \eqref{15}.
\end{IEEEproof}
The distance to the nearest satellite is obtained as a function of system variables such as $N_i$ the number of orbits, $r_i$ the radius of orbits, $M_i$ the number of satellites per orbit, $\phi_i$ inclination angle, $l_u$ receiver latitude, and so on. Using this formula, the behavior of the association distance can be analyzed systematically without system-level simulations. The impact of the orbit radius onto the association distance is examined explicitly, for instance, by taking the derivative of the formula w.r.t. the orbit radius. This result gives an access to the analysis of the association distances which determine the performance of downlink communications from satellites to receivers on Earth.

The derived association distance of the typical receiver at latitude $l_u$ statistically represents the association distance of any receiver at the same latitude $l_u$ since the network architecture is invariant under the Earth rotation \cite{11159552}. Therefore, it corresponds to the spatial average of the association distances collected over all receivers at latitude $l_u.$

It is important to point out that the proposed network model captures a periodically deployed network architecture, leading to structural observations that cannot be identified under nonperiodic models based on binomial or Cox point processes \cite{9079921,9177073,9218989,9497773,9678973}. In models based on binomial, Poisson, or Cox point processes, there is a nonzero probability that a typical receiver observes no satellite, regardless of the number of satellites or the receiver location, which does not align with observations in real-world satellite networks. In contrast, in the proposed model, the distance from a typical receiver to its nearest satellite is always upper bounded by a constant due to the periodic structure, provided that the number of satellites and the number of orbits are sufficiently large and that the receiver is located at a latitude with a visible orbit. This highlights that the proposed multi-altitude Walker point process captures key structural features of real-world satellite networks, providing a meaningful first-order approximation.

\begin{figure}
	\centering
	\includegraphics[width=1\linewidth]{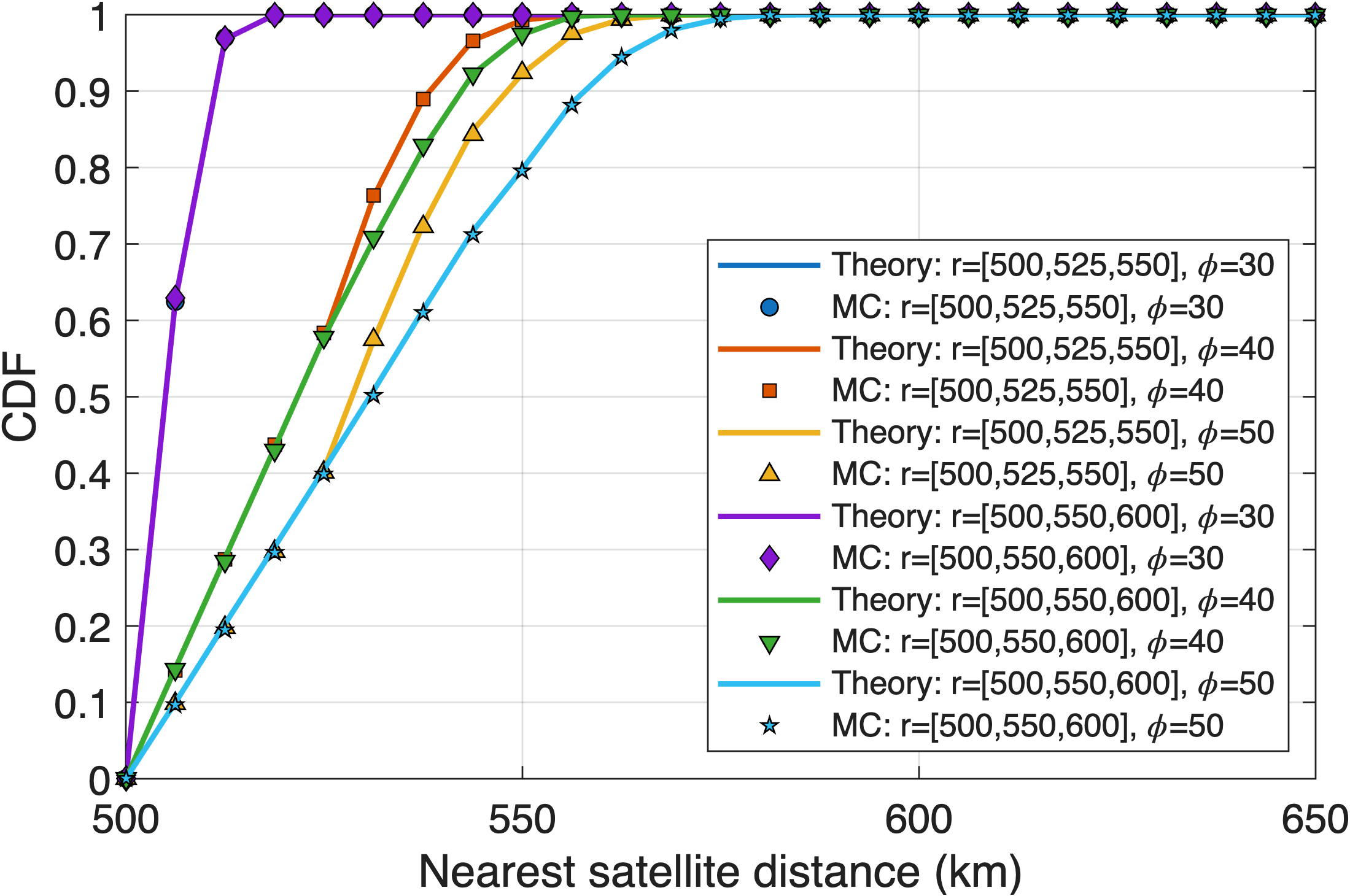}
	\caption{The association distance from typical receiver at latitude $l_u=30^{\circ}$. (Units: degree and km). }
	\label{fig:nearestdistance}
\end{figure}

\begin{figure}
	\centering
		\includegraphics[width=1\linewidth]{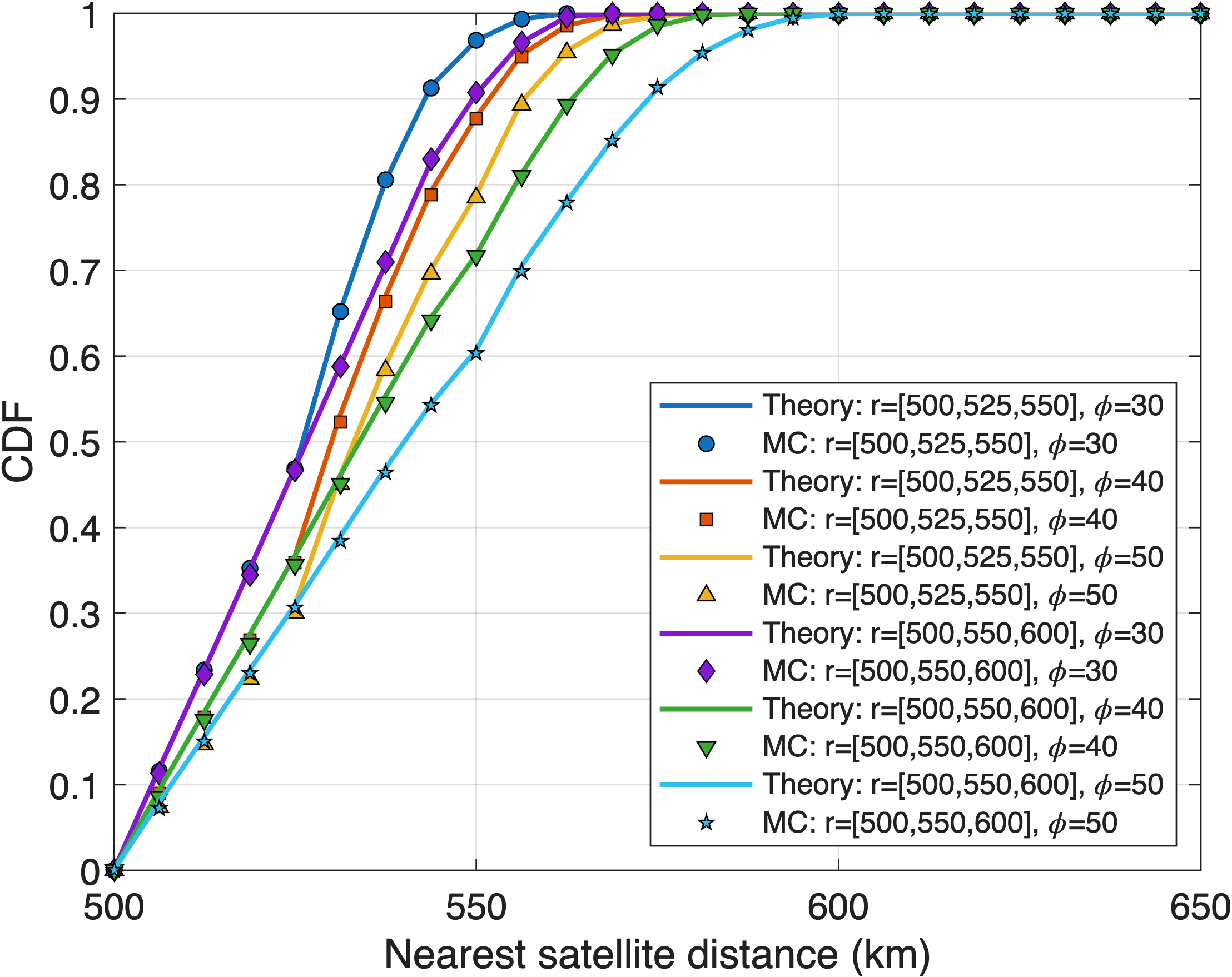}
	\caption{The association distance from typical receiver at latitude $l_u=0^{\circ}$. (Units: degree and km).}
	\label{fig:nearestdistance2}
\end{figure}

Fig. \ref{fig:nearestdistance} illustrates the CDF of the association distance when the typical receiver is at latitude $30^\circ$. We consider three SISA Walker point processes with $I=3$, where $N_1 = N_2 = N_3 = 25$ and $M_1 = M_2 = M_3 = 100$. The orbit inclinations $\phi$ or the satellite altitudes $\{r_1, r_2, r_3\}$ are varied to illustrate the large-scale behavior of the association distance w.r.t. the orbit inclination or satellite altitude. Note that, in Fig. \ref{fig:nearestdistance}, the orbit spacing between different SISA Walker point processes is maximized by setting $\theta_{2,1}=\theta_{1,1}+({2\pi}/{N_1})({1}/{3})$ and $\theta_{3,1}=\theta_{1,1}+({2\pi}/{N_1})({2}/{3})$. The figure shows that when the same numbers of orbits and satellites are used for all three SISA Walkers, a smaller inclination angle leads to a smaller association distance. This is expected because a larger inclination results in satellites being more sparsely distributed in the vertical direction, leading to an increase in the distance to the nearest satellite. Moreover, this experiment shows that, for any inclination angle, the scenario with a wide range of satellite altitudes---namely $r_1=6871$ km , $r_2=6921$ km, $r_3=6971$ km---statistically dominates the scenario with a narrow range of satellite altitudes, namely $r_1=6871$ km, $r_2=6896$ km, $r_3=6921$ km, meaning that the nearest distance under a wide range of altitudes tends to be larger than that under a narrow range. This occurs because a wider range of satellite altitudes results in satellites being more sparsely distributed near the typical receiver, leading to a larger association distance.

Fig. \ref{fig:nearestdistance2} illustrates the distance to the nearest satellite when the typical receiver is located on the equator. Similar to Fig. \ref{fig:nearestdistance}, we consider $I=3$ with $N_1 = N_2 = N_3 = 25$ and $M_1 = M_2 = M_3 = 100$. We also observe that both higher orbit inclinations and wider altitude ranges result in a larger association distance. By roughly comparing Figs. \ref{fig:nearestdistance} and \ref{fig:nearestdistance2}, the association distance generally tends to be larger when the typical receiver is located on the equator. This is mainly because the Walker model with inclined orbits is more densely populated around the inclination latitude than around the equator. In both figures, the derived formula closely matches the Monte Carlo (MC) simulation results, demonstrating the accuracy of the proposed analytical expression.

\begin{proposition}\label{P:1}
	With a slight abuse of notation, let $D_{t}(l_u)$ be the distance from a receiver at $\vec{u}$ to its nearest satellite at time $t\geq 0$, with some arbitrary initial condition. Then, under the ergodicity assumption, we have 
		\begin{align}
		&\lim_{T\to \infty} \frac 1 T \int_0^T \mathbbm{1}(D_{t}(l_u)>d)\  \mathrm{d} t \nnb\\
		&=
		\int \mathbbm{1}(D(l_u)>d) {\overline{\mathcal Q}}_{\Psi}(\mathrm{d} \psi)=\bP(D(l_u)>d),
	\end{align}
where ${\overline{\mathcal Q}}_{\Psi}$ is the unique invariant measure for the proposed satellite dynamic $\{\Psi_t\}$ on $\mathbb{W}.$
\end{proposition}
	\begin{IEEEproof}
	The proof follows from Theorem  \ref{T:rev2:conjugatcy}.
\end{IEEEproof}

To accurately assess the time fraction during which the nearest satellite is at a distance greater than $d$ in practice, system-level simulations must be conducted over a sufficiently long time horizon. In particular, an entire network geometry---including the satellites’ initial locations and their subsequent motion---must be incorporated to obtain reliable estimates of this time fraction. The ergodic interpretation shows that this time fraction can instead be estimated using the closed-form expression provided in Theorem \ref{Theorem:2}, thereby avoiding the need for long-duration system-level simulations.

We now assess how the proposed network geometry impacts the communication performance in practical scenarios by analyzing the coverage probability of the typical receiver. 

\subsection{SINR Coverage Probability}
Let $\bP$ be the probability w.r.t. $\overline{Q}_{\Psi}\times\overline{Q}_{H} $ the product of the invariant measure of the satellite point process and the fading process. The SINR coverage probability of the typical receiver is defined as follow:
\begin{align}
	&\bP(\SINR(l_u)>\tau)
	=\bP\left(\frac{pG_{\star}H\|\vec{X}_\star-\vec{u}\|^{-\alpha}}{\sigma^2+J}>\tau\right) \label{eq:sinrdef},
\end{align}
where $\SINR(l_u)$ is the SINR of the typical receiver at latitude $l_u$ (rad), $\sigma^2$ is the thermal noise power, and $J$ is the interference power seen by the typical receiver defined as follows:
\begin{equation}
	J=\sum_{X_{i,j,k}\in\Psi(\bar{S})\setminus X_\star} \frac{p G_{i,j,k}H_{i,j,k}}{ \|\vec{X}_{i,j,k}-\vec{u}\|^{\alpha}}.\label{eq:intdef}
\end{equation}
In Eqs. \eqref{eq:sinrdef} and \eqref{eq:intdef}, {the $H_{i,j,k}$'s are i.i.d. random variables representing small-scale signal fluctuations independent of everything else}, $X_\star$ is the location of the nearest satellite from the typical receiver, and $G_{i,j,k}$ is the antenna gain between the satellite transmitter $X_{i,j,k}$ and the typical receiver, which is given by $G_{i,j,k} = G(\eta(\|X_{i,j,k}-\vec{u}\|))$. We have 
\begin{align}
	X_\star &= \argmin_{i,j,k\in\Psi(\bar{S})}\left(r_i^2+r_e^2 - 2 (\vec{X}_{i,j,k}\cdot \vec{u})\right),\label{eq:X_star}
\end{align}
where $X_{i,j,k}$ is the location of the $k$-th satellite on the $j$-th orbit of the $i$-th SISA Walker point process. 

Let us denote the set of satellites that are visible to the typical user $\vec{u} $ by  
\begin{align}
	\mathcal{J} &= \left\{ X_{i,j,k}\in\Psi \setminus X_{\star} \left| r_e^2 \leq \left.(\langle \vec{X}_{i,j,k}, \vec{u} \rangle)\right. \right.\right\}.\label{eq:J}
\end{align}
\begin{theorem}\label{Theorem:Coverage}

	When the fading distribution $H$ follows a Gamma random variable with shape parameter $\nu$ and scale parameter $1/\lambda,$ the SINR coverage probability of the typical receiver is given by Eq. \eqref{eq:Theoremcov}. For the special case where $\nu=\lambda=1$, $H$ is an exponential random variable with mean $1$ and the SINR coverage probability of the typical receiver is given by Eq. \eqref{eq:Theoremcov2}. In both equations,  $X_\star$  and $a_{i,j,k}$ are given by Eqs. \eqref{eq:a_ijk} and \eqref{ee:Xstar}, respectively. Note that $X_\star$  and $a_{i,j,k}$ are deterministic functions of $\{\theta,\omega\}$.

\end{theorem}

\begin{figure*}
	\begin{align}
		\bP(\SINR(l_u)>\tau)=&\frac{1}{\Gamma(\nu)}\frac{1}{\frac{2\pi}{N_1}\prod_{i=1}^I(\frac{2\pi}{M_i})}\!\!\int_{0}^{\frac{2\pi}{N_1}}\!\!\!\int_{0}^{\frac{2\pi}{M_1}}\!\!\!\!\cdots\int_{0}^{\frac{2\pi}{M_I}}\!\underbrace{\bE_H\!\cdots\!\bE_H}_{|\mathcal{J}|}\!\left[\Gamma\!\left(\nu,\frac{\lambda\tau(\sigma^2+\sum_{(i,j,k)\in\mathcal J}a_{i,j,k}H_{i,j,k})}{pg_M\|\vec{X}_\star-\vec{u}\|^{-\alpha}}\right)\!\right]\! \underbrace{\diff \omega_{1} \cdots  \diff \omega_{I}}_{I}\diff \theta . \label{eq:Theoremcov}\\
		\bP(\SINR(l_u)>\tau)=&\frac{1}{\frac{2\pi}{N_1}\prod_{i=1}^I(\frac{2\pi}{M_i})}\int_{0}^{\frac{2\pi}{N_1}}\int_{0}^{\frac{2\pi}{M_1}}\cdots\int_{0}^{\frac{2\pi}{M_I}}\left.e^{\frac{- \tau \sigma^2\|X_\star-\vec{u}\|^{\alpha} }{pg_M}}\prod_\mathcal{J} \left(\frac{1}{1+\frac{\tau a_{i,j,k}\|\vec{X}_\star-\vec{u}\|^{\alpha}}{pg_M}}\right)\right. \underbrace{\diff \omega_{1} \cdots  \diff \omega_{I}}_{I}\diff \theta . \label{eq:Theoremcov2}
	\end{align}
	\hrule
\end{figure*}

\begin{IEEEproof}
	Note that $\SINR(l_u)$ is a random variable for the SINR of the typical receiver at latitude $l_u$. For simplicity, we use a shorter notation: $\{{{\theta}},{{\omega}}\}\equiv \{{{\theta}},{{\omega}}_{1},...,{{\omega}}_I\}. $ Then, conditionally on $\{{{\theta}},{{\omega}}\}$, the SINR coverage probability is given by
	\begin{align}
		\bP(\SINR(l_u)>\tau)&=\bE\left[\bP\left(\SINR(l_u)>\tau|\{{{\theta}},{{\omega}}\}\right)\right],\label{SINR_derv1}
	\end{align}
	where the conditional probability inside the expectation is 
	\begin{align}
		&\bP(\SINR(l_u)>\tau|\{{{\theta}},{{\omega}}\})\nnb\\
		&=\bP\left(\left.\frac{pG_\star H\|\vec{X}_\star-\vec{u}\|^{-\alpha}}{\sigma^2+J}>\tau\right|\{{{\theta}},{{\omega}}\}\right)\nnb\\
		&=\bP\left(\left.H> \frac{\tau(\sigma^2+J)\|\vec{X}_\star-\vec{u}\|^\alpha}{p G_{\star} }\right|\{{{\theta}},{{\omega}}\}\right)\nnb\\
		&=\bE\left[\bar{F}_{H}\left(\left.\tau(\sigma^2+J)\|\vec{X}_\star-\vec{u}\|^{\alpha}/(pG_\star)\right)\right|\{{{\theta}},{{\omega}}\}\right]\label{pcov_1}.
	\end{align}

	Furthermore, since the satellite point process $\Psi$ is a function of $\{{{\theta}},{{\omega}}\}$, conditionally on the random variables $\{{{\theta}}, {{\omega}}\},$ the only randomness in the last conditional expression is w.r.t. the $H$'s.  Since the typical receiver is associated with the nearest satellite, we have 
	\begin{align}
		X_\star  &= \argmin_{i,j,k\in\Psi(\bar{S})}\left(\|\vec{X}_{i,j,k}-\vec{u}\|^2\right)\nnb\\
		&=\argmin_{i,j,k\in\Psi(\bar{S})}\left(r_i^2+r_e^2 - 2 \langle \vec{X}_{i,j,k}, \vec{u} \rangle\right) \label{ee:Xstar}
	\end{align}

On the other hand, the interference power at the typical receiver is obtained by combining the received signal powers from all visible satellite transmitters, except the nearest one. Hence, %we have
\begin{align}
	J&=\sum_{\mathcal{J}} \frac{p G_{i,j,k}H_{i,j,k}}{\|\vec{X}_{i,j,k}-\vec{u}\|^{\alpha}}=\sum_{\mathcal{J}}a_{i,j,k}H_{i,j,k},
\end{align}
where 
\begin{equation}
	a_{i,j,k}=pG(\eta(\|\vec{X}_{i,j,k}-\vec{u}\|))/\|\vec{X}_{i,j,k}-\vec{u}\|^{\alpha}. \label{eq:a_ijk}
\end{equation}

To simplify, we use the projection of vector $\vec{X}_{i,j,k}$ onto vector $\vec{u}$ to obtain the following expression:
\begin{align}
	\mathcal{J}&=\{X_{i,j,k}\in \Psi(\bar{S})\} \setminus X_{\star}\nnb\\
	&\equiv \left\{ X_{i,j,k}\in\Psi \left| r_e^2 \leq \left.\langle \vec{X}_{i,j,k}, \vec{u} \rangle\right. \right.\right\}\setminus X_{\star}.\nnb
\end{align}

Now, let us assume that $H$ is a Gamma random variable with shape parameter ${\nu}$ and scale parameter ${1/\lambda},$ denoted by $H\sim\mathrm{Gamma}(\nu,1/\lambda)$. Let $\Gamma(\nu)$ be the ordinary Gamma function defined by $\int_0^\infty t^{\nu-1}e^{-t}\diff t. $ Then, the PDF and CDF of $H$ are
\begin{align*}
	f_H(x)= \frac{\lambda^{\nu}x^{\nu-1}e^{-\lambda x}}{\Gamma(\nu)}, \quad {F}_H(x)=\frac{1}{\Gamma(\nu)}\gamma(\nu,\lambda x),  % \int_x^\infty f_H(u)\diff u = \gamma(x,n)\nnb,
 \end{align*}
respectively, where $\gamma(\nu,\lambda x)=\int_0^{\lambda x} t^{\nu-1}e^{-t}\diff t  $ is the lower incomplete Gamma function. We also have 
\begin{equation}
	\bar{F}_H(x)=\frac{1}{\Gamma(\nu)}\Gamma(\nu,\lambda x), \label{eq:bar_F_H}
\end{equation}
where $\Gamma(\nu,\lambda x) =  \int_{\lambda x}^\infty t^{\nu-1}e^{-t}\diff t  $ is the upper incomplete Gamma function. By using the CCDF in Eq. \eqref{eq:bar_F_H}, the coverage probability is given by
\begin{align}
	&\bP(\SINR(l_u)>\tau|\{{{\theta}},{{\omega}}\}) \nnb\\
	&=\frac{1}{\Gamma(\nu)}\underbrace{\bE_{H}\!\cdots\!\bE_{H}}_{|\mathcal{J}|}\left[\left.\Gamma\!\left(\!\nu,\frac{\lambda\tau(\sigma^2+J)\|\vec{X}_\star-\vec{u}\|^{\alpha}}{pg_M}\right)\!\right|\{{{\theta}},{{\omega}}\}\right]\nnb\\
	&=\frac{1}{\Gamma(\nu)}\underbrace{\bE_{H}\!\cdots\!\bE_{H}}_{|\mathcal{J}|}\!\left[\Gamma\!\left(\nu,\frac{\lambda\tau(\sigma^2+\sum\limits_{(i,j,k)\in\mathcal J} \!\!\!a_{i,j,k}H_{i,j,k})}{pg_M\|\vec{X}_\star-\vec{u}\|^{-\alpha}}\right)\right].\nnb
\end{align}
 Finally, by integrating the last expression w.r.t. $\{{{\theta}},{{\omega}}\}$, or equivalently $\overline{\mathcal{Q}}$, we obtain the final expression for the coverage probability of the typical receiver.

Now, let us consider the special case of the Gamma-distributed fading when $\nu=\lambda=1$. For this case, this accounts for the Rayleigh fading scenario and $H$ follows an exponential random variable with mean $1$. The PDF and CCDF of $H$ are given by $f_H(x)=e^{-x}, \ \bar{F}_H(x) = e^{-x},$ respectively.
To obtain the coverage probability of the typical receiver, we first derive the Laplace transform of the interference $J$ conditionally on $\{{{\theta}},{{\omega}}\}$ as follows:
\begin{align}
	\cL_J(s)&=\bE\left[\left.e^{-Js}\right|\{\theta,\omega\} \right]\nnb\\
	&=\bE\left[\left.\prod_{\mathcal J}e^{-spG_{}H\|\vec{X}-\vec{u}\|^{-\alpha}}\right|\{\theta,\omega\}\right]\nnb\\
	&=\left.\prod_\mathcal{J} \cL_H\left(\frac{spG_{i,j,k}}{\|\vec{X}_{i,j,k}-\vec{u}\|^{\alpha}}\right)\right.=\left.\prod_\mathcal{J} \frac{1}{1+sa_{i,j,k}}\right.\label{eq:Lplace}.
\end{align}
Let us denote by $\cL_H(s)$ the Laplace transform of $H$ evaluated at $s$. Using the last CCDF, Eq. \eqref{SINR_derv1} is given by 
\begin{align}
	&\bP(\SINR(l_u)>\tau|\{{{\theta}},{{\omega}}\})\nnb\\
	&= \bE\left[\left.\exp\left(-\frac{\tau(\sigma^2+J)\|\vec{X}_\star-\vec{u}\|^{\alpha}}{pg_M}\right)\right| \{\theta,\omega\}\right]\nnb\\
&= e^{\frac{- \tau \sigma^2\|X_\star-\vec{u}\|^{\alpha} }{pg_M}}\cL_J\left(\frac{\tau\|\vec{X}_\star-\vec{u}\|^{\alpha}}{pg_M}\right)\nnb\\
&=e^{\frac{- \tau \sigma^2\|X_\star-\vec{u}\|^{\alpha} }{pg_M}}\left.\prod_\mathcal{J} \left(\frac{1}{1+\frac{\tau a_{i,j,k} \|\vec{X}_\star-\vec{u}\|^{\alpha}}{pg_{M}}}\right)\right.\label{pcov-3},
\end{align}
where $\cL_J(x)$ denotes the Laplace transform of the interference $J$ evaluated at $x$, which is given by Eq. \eqref{eq:Lplace}. Finally, by integrating the above expression w.r.t. the unique invariant distribution $\overline{\mathcal{Q}}$ on $\mathbb S$, we obtain the final expression of the coverage probability of the typical receiver at latitude $l_u$ when $H$ is an exponential random variable with mean $1.$
\end{IEEEproof}

\begin{figure}
	\centering
	\includegraphics[width=1\linewidth]{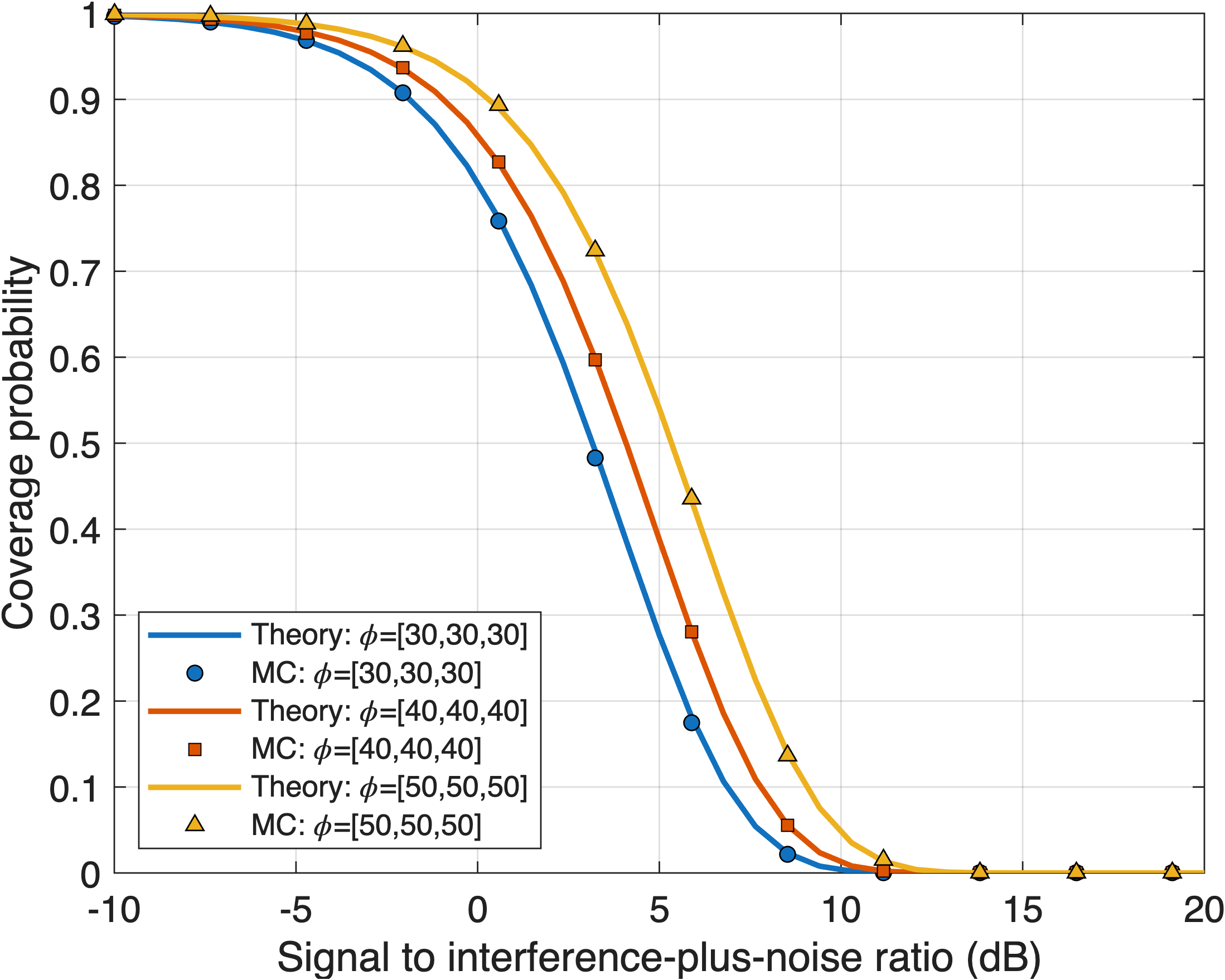}
	\caption{The SIR coverage probability of the typical receiver at $l_u=30^{\circ}$.}
	\label{fig:sinrvalidationfig1}
\end{figure}
\begin{figure}
	\centering
	\includegraphics[width=1\linewidth]{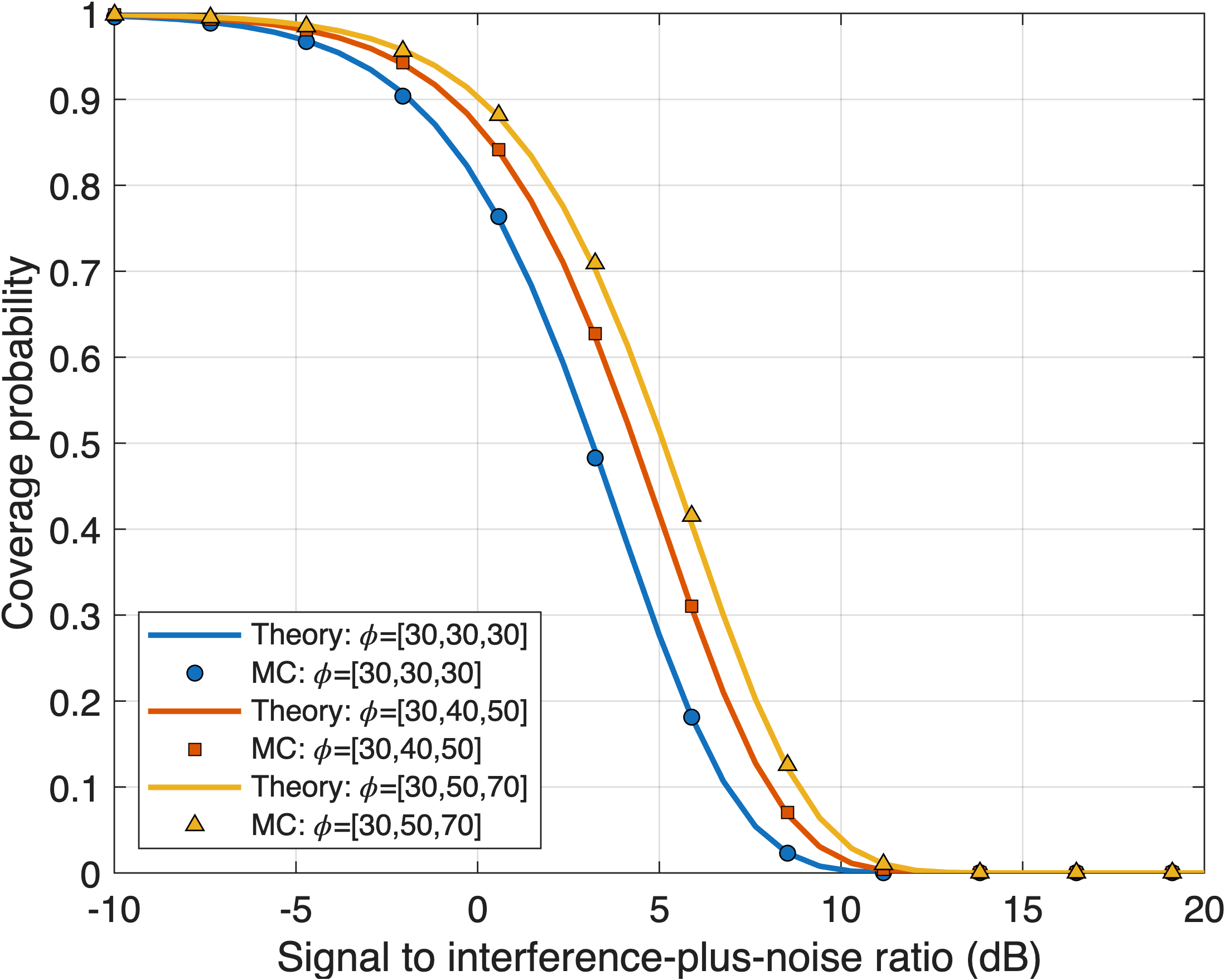}
	\caption{The SINR coverage probability of the typical receiver at $l_u=30^{\circ}$.}
	\label{fig:sinrvalidationfig2}
\end{figure}

\begin{figure}
	\centering
	\includegraphics[width=1\linewidth]{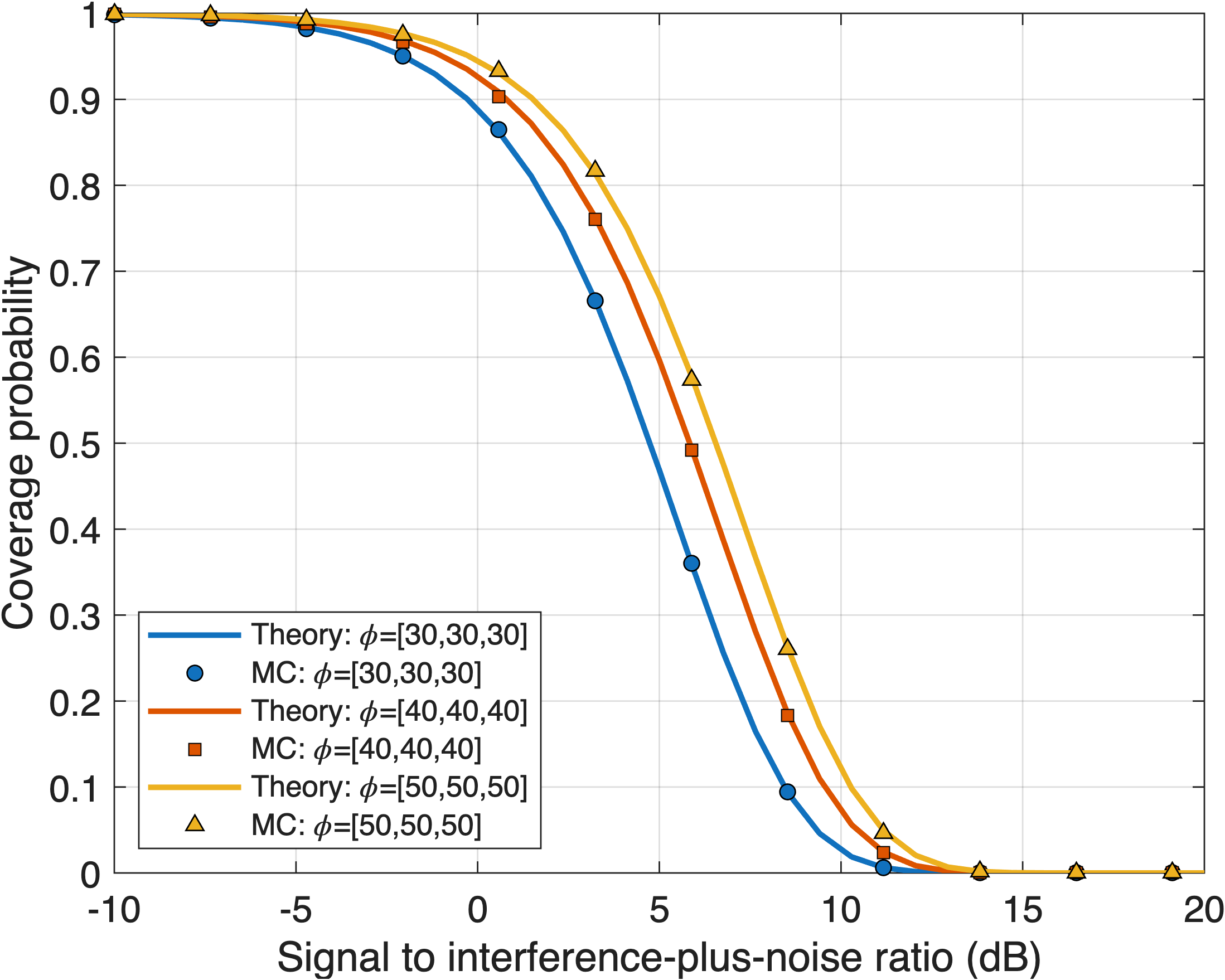}
	\caption{The SINR coverage probability of the typical receiver at $l_u=0^{\circ}$.}
	\label{fig:sinrvalidationfig3}
\end{figure}

\begin{figure}
	\centering
	\includegraphics[width=1\linewidth]{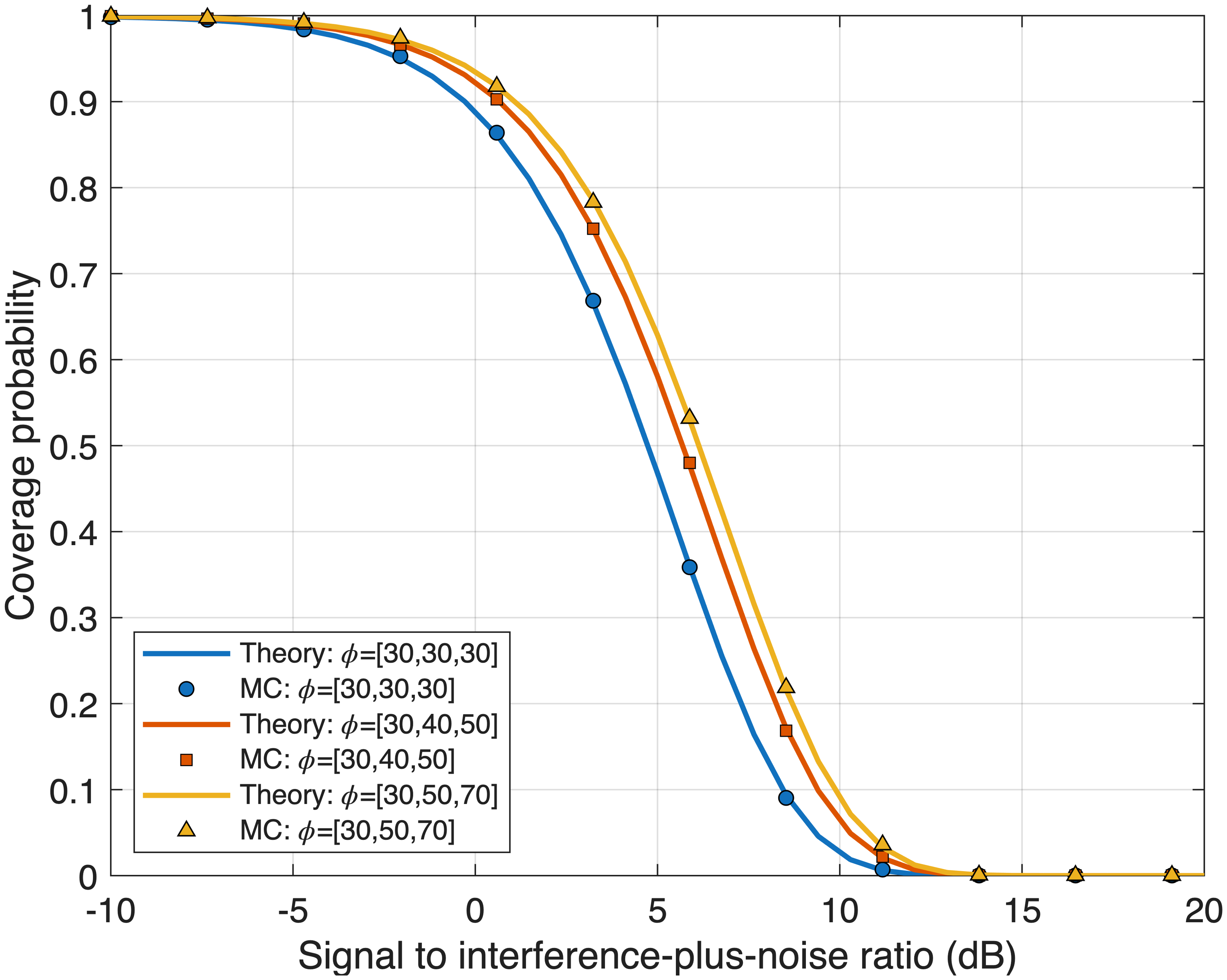}
	\caption{The SINR coverage probability of the typical receiver at $l_u=0^{\circ}$.}
	\label{fig:sinrvalidationfig4}
\end{figure}

\begin{table}
	\caption{System parameters}\label{Table:1}
	\centering
	\begin{tabular}{|c|c|}
		\hline
		Parameters & Values \\
		\hline
		1-m reference power 	$p$ (Watt) & $100 $ dBW  \\
		\hline
		Bandwidth	&  $10$ MHz \\
		\hline
		Noise figure $N_f$	&   $7$ dB\\
		\hline
		Noise power $\sigma^2$ (Watt) & $-127$ dBW\\
		\hline
		Temperature $T $	& $300$ K \\
		\hline
		Max transmit antenna $g_M$	& $26$ dBi \\
		\hline
		Transmit antenna gain $g_m$	&  $10 $ dBi \\
		\hline
		Cut-off angle $\eta_c $ & $\approxeq 11^{\circ}$ \\
		\hline
		Cut-off distance $d$ such that $\eta(d)=\eta_c$ & 510 km\\
		\hline   
		Total number of SISA Walker point $(I)$ & 3\\
		\hline
		Number of orbits $N_1,N_2,N_3$ & 50, 50, 50 \\
			\hline
			Number of satellites per orbit $M_1,M_2,M_3$ & 100, 100, 100\\
			\hline
			Orbit radii $r_1,r_2,r_3$ (km) & $6871$, $6896$, $6921$ \\
			\hline
			Path loss exponent & $\alpha=2$\\
			\hline 
	\end{tabular}
\end{table}
Figures \ref{fig:sinrvalidationfig1}--\ref{fig:sinrvalidationfig4} validate the analytical SINR coverage expression for the case of Gamma fading with shape parameter $\nu=2$ and scale parameter $1/2$, corresponding to unit-mean fading. A common observation across all four figures is that the theoretical curves closely match MC simulation results over the entire SINR range. The agreement remains accurate not only in the moderate coverage region but also in the low and high SINR regimes, confirming the correctness of the derived coverage formula and its numerical implementation. In addition, all curves exhibit the expected sigmoidal behavior, where the coverage probability gradually decreases as the SINR threshold increases. The close overlap between the analytical predictions and simulation markers demonstrates that the proposed framework accurately captures both the serving-satellite geometry and the aggregate interference generated by the heterogeneous Walker constellation.

Considering each figure individually, Figs. \ref{fig:sinrvalidationfig1} and \ref{fig:sinrvalidationfig3} investigate homogeneous inclination configurations, where all three SISA layers share the same inclination angle. In both figures, increasing the inclination angle from $30^\circ$ to $50^\circ$ shifts the coverage curves toward higher SINR thresholds, indicating improved coverage performance. This improvement suggests that larger inclination angles provide more favorable satellite visibility and serving-link geometry for the considered users. Similarly, Figs. \ref{fig:sinrvalidationfig2} and \ref{fig:sinrvalidationfig4} examine heterogeneous inclination configurations. As the inclination set changes from $[30^\circ,30^\circ,30^\circ]$ to $[30^\circ,50^\circ,70^\circ]$, the coverage probability consistently increases for a fixed SINR threshold. This trend indicates that introducing constellation diversity through multiple inclination angles can improve the probability of establishing stronger serving links while maintaining manageable interference levels.

Comparing Figs. \ref{fig:sinrvalidationfig1}--\ref{fig:sinrvalidationfig4} provides additional geometric insight into the impact of user latitude and orbital inclination diversity. First, the relative ordering of the curves remains unchanged across all figures, suggesting that the inclination structure is a dominant factor governing the SINR distribution. Second, the performance gain obtained by increasing the inclination angles is observed for both the equatorial user and the user located at latitude $30^\circ$, indicating that the benefit is robust w.r.t user location. Third, the heterogeneous inclination configurations considered in Figs. \ref{fig:sinrvalidationfig2} and \ref{fig:sinrvalidationfig4} generally achieve performance comparable to or better than their homogeneous counterparts, suggesting that combining multiple orbital inclinations can improve spatial coverage uniformity. Overall, the four figures collectively demonstrate that increasing orbital inclination and introducing inclination diversity both contribute to improved downlink SINR coverage while preserving excellent agreement between theory and simulation.

\begin{figure}
	\centering
	\includegraphics[width=1\linewidth]{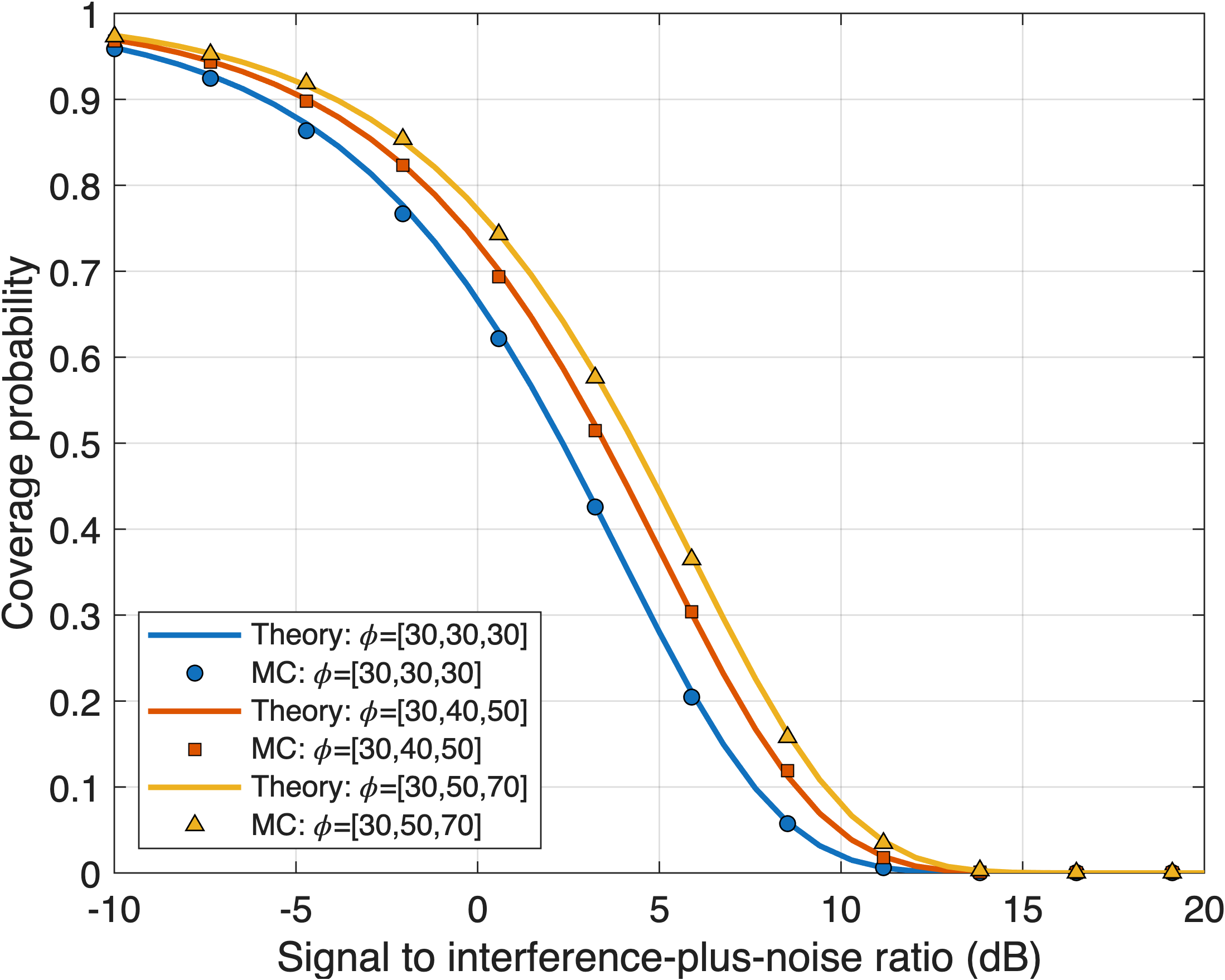}
	\caption{The SINR coverage probability of the typical receiver at $l_u=0^{\circ}$.}
	\label{fig:sinrvalidationfig6}
\end{figure}

\begin{figure}
	\centering
	\includegraphics[width=1\linewidth]{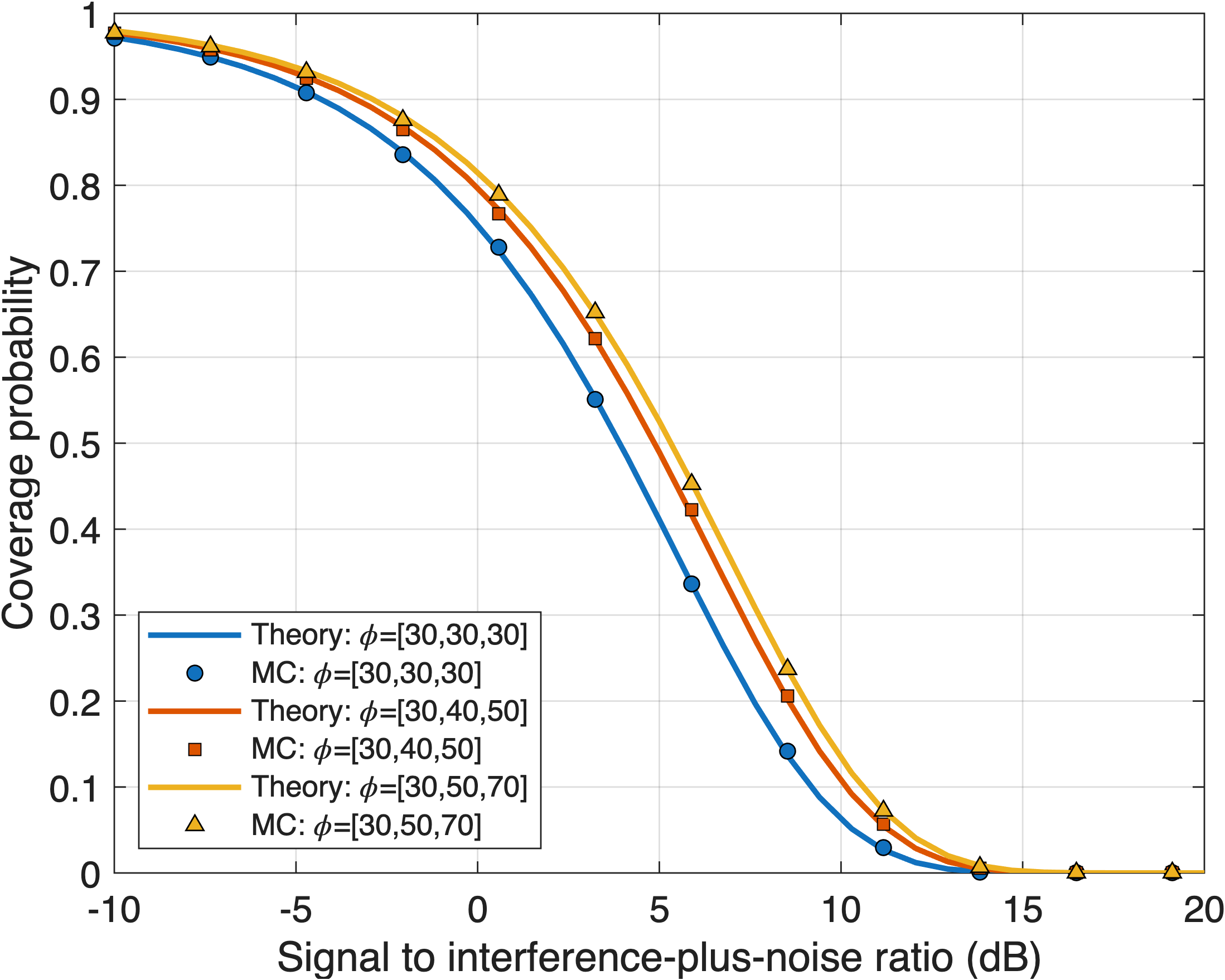}
	\caption{The SINR coverage probability of the typical receiver at $l_u=0^{\circ}$.}
	\label{fig:sinrvalidationfig8}
\end{figure}

Figs. \ref{fig:sinrvalidationfig6} and \ref{fig:sinrvalidationfig8} illustrates the SINR coverage probability of the typical user at latitude $30^{\circ}$ and $0^{\circ}, $ respectively, when $H$ is an exponential random variable. 
An interesting observation from Figs. \ref{fig:sinrvalidationfig6} and \ref{fig:sinrvalidationfig8} is that the qualitative behavior remains nearly unchanged compared to the Gamma fading case shown in Figs. \ref{fig:sinrvalidationfig2} and \ref{fig:sinrvalidationfig4}. In particular, increasing the diversity of orbital inclinations consistently improves the SINR coverage probability, while the relative gap between the three configurations remains similar. This indicates that the dominant factor governing the coverage performance is the constellation geometry rather than the specific fading model. Consequently, the benefits of inclination diversity appear to be robust across different channel fading assumptions.

\begin{proposition}\label{P:2}
	Let us further define by $\SINR(l_u,t;x)$ the SINR of the typical receiver at time $t$ with the initial condition $x\in \bS$ as follows: 
	\begin{equation}
		\SINR(l_u,t;x) = \frac{p g_MH_{\star}(t)\|\vec{X}_{\star}-\vec{u}_0\|^{-\alpha}}{\sum\limits_{X_i\in\Psi_{t;x}(\bar{S})\setminus X_\star} \frac{p G_iH_i(t)}{\|\vec{X}_{i}-\vec{u}\|^{\alpha}}+\sigma^2}, \label{eq:40}
	\end{equation}
	where $\Psi_{t;x}$ denotes the satellite point process at time $t$ for the initial condition $x$. Assume the fading processes $H_i(\cdot) $ are stationary and mixing processes, independent of everything else.
	
	Then, if speeds are rationally independent, for $\overline{\mathcal Q}$-almost all initial conditions $x\in\mathbb S$ and for almost every fading sample path, we have
		\begin{align}
			\lim_{T\to \infty} \frac 1 T \int_0^T \mathbbm{1}(\SINR(l_u,t;x)>\tau)\!\diff  t 
			=\bP(\SINR(l_u)>\tau).\label{eq:prop2}
		\end{align}
\end{proposition}
\begin{IEEEproof}If speeds are rationally independent and $H(\cdot)$ is an independent and mixing process, the dynamics $\Psi_{t}$ and $H(t)$ are jointly stationary and ergodic ergodic w.r.t. their invariant measure $\overline{\mathcal{Q}}_{\Psi}\times \overline {\mathcal{Q}}_H$.  This completes the proof.
\end{IEEEproof}
Note that the first term in Eq. \eqref{eq:prop2} is the  long-term time average of the indicator function that takes the value of $1$ if the SINR of the typical receiver at a given time is greater than $\tau.$ This term practically accounts for the time fraction of the typical receiver over which its SINR is greater than $\tau$. Then, the last Proposition theoretically proves that the time fraction of SINR coverage of a typical receiver is in fact the same as the probability that the SINR of the typical receiver is greater than a given threshold $\tau$, averaged w.r.t. the invariant distribution ${\overline{\mathcal Q}}_{\Psi}$. This holds for all initial conditions.

\subsection{Ergodic Receiver Throughput}
The rate of the typical receiver at latitude $l_u$ is defined by the ergodic capacity of the typical receiver in the downlink communication as follows:
\begin{equation}
	\cR = \bE\left[\log_2\left(1+\SINR(l_u)\right)\right]\label{eq:def},
\end{equation}
where $\SINR(l_u)$ denotes the positive random variable for the SINR of the typical receiver at latitude $l_u.$ The rate of the typical receiver represent the average amount of total bits/sec/Hz in the downlink communication from satellite transmitters to receivers on Earth.
\begin{theorem}
	When $H$ is a Gamma random variable with shape parameter $\nu$ and scale parameter $1/\lambda,$  the ergodic throughput of the typical receiver at latitude $l_u$ is given by Eq. \eqref{eq:Theorem_rate} with $X_\star$ and $\mathcal{J}$ given by Eqs. \eqref{eq:X_star} and \eqref{eq:J}, respectively. 
	
	For the special case when the fading distribution $H$ follows an exponential random variable with mean one, the ergodic throughput of the typical receiver at latitude $l_u$ is given by Eq. \eqref{eq:Theorem_rate2}.
\end{theorem}
\begin{IEEEproof}
	Inside the expectation in Eq \eqref{eq:def}, the expression $\log_2(1+\SINR)$ is a random variable with its CCDF given by $\bP(\log_2(1+\SINR)>u) = \bP(\SINR>2^{u}-1) $, the coverage probability evaluated at $2^u-1$ with dummy variable $u$. Since the rate is the expectation of the positive random variable $\log_2(1+\SINR)$, and the mean of a positive random variable $X$ is evaluated by $\int_0^{\infty}\bP(X>x)\diff x $, we have $\mathcal{R} 
	=\int_{0}^{\infty}\bP(\SINR\geq 2^{u}-1)\diff u.\nnb$
	To get the final result, we use the coverage probability given by Eq. \eqref{eq:Theoremcov} or Eq. \eqref{eq:Theoremcov2} to obtain the ergodic rate of the typical receiver for corresponding fading distribution.
\end{IEEEproof}

\begin{proposition}\label{P:3}
	For $\overline{\mathcal Q}$-almost every initial condition $x\in\mathbb S$ and for almost all fading sample path, the average throughput of the typical receiver at latitude $l_u$ satisfies
	\begin{multline}
		\bE\left[\log_2\left(1+\SINR(l_u)\right)\right]\\
		=\lim_{T\to\infty}\frac1T
		\int_0^T
\left.\log_2\left(1+\SINR(l_u,t;x)\right)\right.
		\diff t,
	\end{multline}
	where $\SINR(l_u,t;x)$ is defined in Eq. \eqref{eq:40}. 
\end{proposition}

\begin{IEEEproof} The proof follows from \ref{P:2}
\end{IEEEproof}

\begin{figure*}
	\begin{align}
		\cR=
		&\frac{1}{\Gamma(\nu)}\frac{1}{\left.\frac{2\pi}{N_1}\prod_{i=1}^I \frac{2\pi}{M_i}\right.}\int_0^{\infty}\int_{0}^{\frac{2\pi}{N_1}}\!\int_{0}^{\frac{2\pi}{M_1}}\!\cdots\int_{0}^{\frac{2\pi}{M_I}}\!\left.\underbrace{\bE_H\cdots\bE_H}_{|\mathcal{J}|}\left[\Gamma\left(\nu,\frac{(2^{y}-1)(\sigma^2+\sum\limits_{(i,j,k)\in\mathcal J}a_{i,j,k}H_{i,j,k})}{pg_M\|\vec{X}_\star-\vec{u}\|^{-\alpha}}\right)\right]\right.\underbrace{\diff \omega_{1}  \cdots \diff \omega_{I}}_{I} \diff \theta \diff y, \label{eq:Theorem_rate}\\
		\cR=
		&\frac{1}{\left.\frac{2\pi}{N_1}\prod_{i=1}^I \frac{2\pi}{M_i}\right.}\int_{0}^{\infty}\!\int_{0}^{\frac{2\pi}{N_1}}\!\int_{0}^{\frac{2\pi}{M_1}}\!\cdots \!\int_{0}^{\frac{2\pi}{M_I}}\left(e^{\frac{- (2^{y}-1) \sigma^2\|X_\star-\vec{u}\|^{\alpha} }{pg_M}}\prod_\mathcal{J} \left(\frac{1}{1+\frac{(2^{y}-1) a_{i,j,k}}{pg_M \|\vec{X}_\star-\vec{u}\|^{\alpha}}}\right)\right)\underbrace{\diff \omega_{1} \cdots \diff \omega_{I}}_{I} \diff \theta \diff y . \label{eq:Theorem_rate2}
	\end{align}
				\begin{align}
		\cL_T(s)=
		&\frac{1}{\left.\frac{2\pi}{N_i}\prod_{i=1}^I \frac{2\pi}{M_i}\right.}\!\int_{0}^{\frac{2\pi}{N_1}}\!\int_{0}^{\frac{2\pi}{M_1}}\!\cdots\!\int_{0}^{\frac{2\pi}{M_I}}\prod_{\bar{\mathcal J}} \left(1+{ sa_{i,j,k}}/{\lambda}\right)^{-\nu}\underbrace{\diff \omega_{1}  \cdots \diff \omega_{I}}_{I}\diff \theta , \label{eq:Theorem_tot}\\
		\cL_T(s)=
		&\frac{1}{\left.\frac{2\pi}{N_i}\prod_{i=1}^I \frac{2\pi}{M_i}\right.}\int_{0}^{\frac{2\pi}{N_1}}\int_{0}^{\frac{2\pi}{M_1}}\cdots\!\int_{0}^{\frac{2\pi}{M_I}}\left. \prod_{\bar{\mathcal J}} \left(1+{sa_{i,j,k}}\right) \right.\underbrace{\diff \omega_{1}  \cdots \diff \omega_{I}}_{I} \diff \theta . \label{eq:Theorem_tot2}
	\end{align}
	\hrule
\end{figure*}

\subsection{Total Received Power}
The total received power gives the total amount of  energy received by a random location of latitude $l_u$. This metric indicates the interference incur by the Walker satellites on orbits and therefore will be used to directly measure the impact of satellite downlink communications to non-terrestrial receivers such as mobile cellular handsets or vehicle transceivers, provided that they are unable to communicate with satellites. Let $T$ denote the total received  power at the typical receiver and it is defined as
\begin{equation}
	T = \sum_{X_{i,j,k}\in\Psi(\bar{S})} \frac{p G_{i,j,k}H_{i,j,k}}{\|\vec{X}_{i,j,k}-\vec{u}\|^{\alpha}}.
\end{equation}
The total received signal power is a random variable and we characterize its distribution by deriving its Laplace transform.
\begin{theorem}
	If $H$ follows a Gamma distribution with parameters $\nu$ and $1/\lambda,$ the Laplace transform of the total received signal power is given by Eq. \eqref{eq:Theorem_tot}.

	On the other hand, if $H$ follows an exponential random variable with mean one, the Laplace transform of the total received signal power is given by Eq. \eqref{eq:Theorem_tot2}.
\end{theorem}

\begin{figure*}

	\hrule
\end{figure*}

\begin{IEEEproof}
	Let $\overline{\mathcal{J}}$ be the set of visible satellite defined as follows: \begin{align}
		\bar{\mathcal{J}} %&= \{X_{i,j,k}\in\Psi(\bar{S})| \}\nnb\\
		&=\left\{ X_{i,j,k}\in\Psi \left. \text{ such that }r_e^2 \leq \left.\langle \vec{X}_{i,j,k}, \vec{u} \rangle\right. \right.\right\}\nnb.
	\end{align}
Conditionally on  $\{{{\theta}}, {{\omega}}\}$, the Laplace transform of the total received signal power is given by
	\begin{align}
		\cL_T(s) %&= \bE\left[\exp\left(-sT\right)\right]\nnb\\
				 &= \bE\left[\prod_{ \bar{\mathcal{J}}} \exp\left(\frac{-spG_{i,j,k}H_{i,j,k}}{\|\vec{X}_{i,j,k}-\vec{u}\|^\alpha}\right)\right]\nnb\\
	&=\bE\left[\left.\prod_{\bar{\mathcal{J}}}\cL_H\left(-sa_{i,j,k}\right)\right|\{{{\theta}},{{\omega}}\}\right] \label{eq:Totalrxpower},
	\end{align}
	To get Eq. \eqref{eq:Totalrxpower}, we first use the law of total probability and then employ the fact that conditionally on $\{\bar \theta,{{\omega}}\}$, the Laplace transform of $T$ is given by the product of Laplace transforms of $H$ evaluated at points that are given by measurable functions of $\{\bar \theta,{{\omega}}\}$. Then, employing the Laplace transform expression for various $H$ gives the final result.
\end{IEEEproof}

\section{Conclusion}
In this paper, we develop a stochastic geometry based dynamical system framework for modeling and analyzing heterogeneous satellite networks with multi-altitude orbits. By representing satellite networks as a superposition of multiple Walker point processes, we characterize both the spatial network layout and its temporal evolution. Within this framework, we derive a universal invariant measure and the exact conditions under which the proposed satellite system is ergodic, thereby establishing a basis for estimating long-term time averages via spatial averages w.r.t. the invariant measure. In the ergodic regime, we further derive the nearest-satellite distance distribution, the SINR coverage probability under various fading models, the aggregate interference, and the ergodic rate. Taken together, the proposed framework provides a tractable tool for evaluating downlink performance in heterogeneous LEO satellite networks with multi-altitude orbits.

The framework also opens several directions for future work. One example is satellite-to-satellite routing, where the temporal structure of the constellation determines the local delay, the number of potential relay satellites, and the total routing distance. The framework developed in this work enables the investigation of the temporal evolution of inter-satellite routing in a natural manner. In addition, the proposed model can be combined with isotropic Cox-based constructions to describe hybrid network architectures in which regular and random orbits coexist. Finally, the analytical tools presented here can be extended to analyze and optimize various interference mitigation and resource management algorithms in next-generation satellite networks.

\section*{Acknowledgment}
The work of Chang-Sik Choi was supported by NRF RS-2024-00334240.
The work of Francois Baccelli was supported by the European Research Council NEMO project (grant ERC 788851), the Horizon Europe INSTINCT project (grant SNS 101139161), the France 2030 projects PEPR réseaux du Futur project  (grant ANR-22-PEFT-0010), and by 5G NTN mmWave (BPIFrance). This joint work was also supported by a South Korea--France Hubert Curien grant.

\bibliographystyle{IEEEtran}
\bibliography{ref2}

\end{document}